\documentclass[11pt, fleqn]{article}

\usepackage[letterpaper, margin=1in]{geometry}
\usepackage{graphicx}
\usepackage{authblk}
\usepackage{hyperref} 
\usepackage{amsmath}
\usepackage{amssymb}
\usepackage{amsfonts}
\usepackage{amsthm}
\usepackage{fullpage}
\usepackage{enumitem}
\usepackage{algorithm}
\usepackage[noend]{algpseudocode}
\let\oldReturn\Return
\renewcommand{\Return}{\State\oldReturn}
\usepackage{xcolor}
\usepackage{cite}

\newtheorem{theorem}{Theorem}[section]
\newtheorem{lemma}[theorem]{Lemma}

\newtheorem{observation}[theorem]{Observation}

\newtheorem{remark}[theorem]{Remark}

\newcommand{\bx}{\mathbf{x}} 
\newcommand{\by}{\mathbf{y}}

\newcommand{\bc}{\mathbf{c}}
\newcommand{\bw}{\mathbf{w}}
\newcommand{\bp}{\mathbf{p}}

\let\epsilon\varepsilon

\title{Existence of Fair and Efficient Allocation of Indivisible Chores}
\author{Ryoga Mahara}
\affil{Department of Mathematical Informatics \\ The University of Tokyo \\ \texttt{mahara@mist.i.u-tokyo.ac.jp}}
\date{}

\begin{document}

\maketitle
\begin{abstract}
We study the problem of allocating indivisible chores among agents with additive cost functions in a {\it fair} and {\it efficient} manner. 
A major open question in this area is whether there always exists an allocation that is envy-free up to one chore (EF1) and Pareto optimal (PO).
Our main contribution is to provide a positive answer to this question by proving the existence of such an allocation for indivisible chores under additive cost functions.
This is achieved by a novel combination of a fixed point argument and a discrete algorithm, providing a significant methodological advance in this area.

Our additional key contributions are as follows.
We show that there always exists an allocation that is EF1 and fractional Pareto optimal (fPO), where fPO is a stronger efficiency concept than PO.
We also show that an EF1 and PO allocation can be computed in polynomial time when the number of agents is constant.
Finally, we extend all of these results to the more general setting of weighted EF1 (wEF1), which accounts for the entitlements of agents.
\end{abstract}

\newpage

\section{Introduction}\label{sec:int}

The fair allocation of \emph{indivisible} resources and burdens is a fundamental problem that has been extensively studied in mathematics, economics, and computer science~\cite{brams1996fair,brandt2016handbook}, and it has received increasing attention in recent years (see surveys~\cite{amanatidis2023fair,aziz2022algorithmic,guo2023survey}).
Fair allocation of indivisible items arises in many practical settings, such as inheritance division, course allocation, and task assignment.
In this work, our goal is to allocate a finite set of items among agents with individual valuation functions in a {\it fair} and {\it efficient} manner.
We assume that each agent has an additive valuation function.
We call an item a \emph{good} if all agents assign it a non-negative value, and a \emph{chore} if all agents assign it a non-positive value.

\emph{Envy-freeness} (EF)\cite{foley1966resource} is a fundamental notion of fairness. An allocation is EF if each agent prefers their own bundle over the bundle of any other agent.
However, when items are indivisible, an EF allocation may not exist even with just one item and two agents.
To overcome this limitation, several relaxed fairness concepts have been proposed. Among them, the most widely used is \emph{envy-freeness up to one item} (EF1)\cite{budish2011combinatorial}.
In the case of goods, EF1 requires that each agent prefers their own bundle to that of any other agent after removing at most one item from the other agent’s bundle, whereas for chores, it requires that each agent prefers their own bundle to that of any other agent after removing at most one item from their own bundle.
An EF1 allocation exists and is computed in polynomial time even for general combinatorial valuation functions, not just additive ones~\cite{lipton2004approximately,bhaskar2021approximate}.

\emph{Pareto optimality} (PO) is a standard notion of efficiency. An allocation is PO if no agent can be made better off without making someone else worse off.
A stronger notion, called \emph{fractional Pareto optimality} (fPO), extends this idea to fractional combinations of allocations.
Classical results in economic theory~\cite{negishi1960welfare, varian1976} show a connection between fractional Pareto-optimal allocations and weighted utilitarian welfare.
This connection immediately implies both the existence and efficient computability of fPO allocations, consequently, of PO allocations as well.

Achieving both fairness and efficiency simultaneously is a natural and fundamental objective in the fair division. 
The existence of EF1 or PO allocations is relatively easy to guarantee, and each can be computed separately without much difficulty.  
However, finding an allocation that satisfies both properties simultaneously, or even proving that such an allocation exists, is a natural and important goal but also a highly challenging problem.
The landscape of prior work differs depending on the type of items we consider.

\paragraph{Fair and efficient allocation for indivisible goods}
Caragiannis et al.~\cite{caragiannis2019unreasonable} established the novel result that maximizing the Nash social welfare~\cite{Nash50, kaneko1979nash}, defined as the geometric mean of the agents’ valuations, results in an allocation that is both EF1 and PO. However, since maximizing Nash social welfare is NP-hard~\cite{nguyen2014computational} and even APX-hard~\cite{lee2017apx}, this approach does not directly yield an efficient algorithm.
To overcome this barrier, Barman et al.~\cite{barman2018finding} proposed a pseudo-polynomial time algorithm that computes an EF1 and PO allocation and proved the existence of an EF1 and fPO allocation.
Nevertheless, their existence proof relies on a non-constructive convergence argument and does not directly yield an algorithm for computing such an allocation.
Later, Mahara~\cite{mahara2024polynomial} developed a polynomial-time algorithm for computing an EF1 and fPO allocation when the number of agents is constant.
The existence of a polynomial-time algorithm for finding an EF1 and PO (or fPO) allocation remains an important open question.

\paragraph{Fair and efficient allocation for indivisible chores}
Compared to the case of goods, the fair division of chores is much more challenging, even though the two settings look similar.
Despite many efforts by researchers, progress on chore division has been relatively limited.
The existence of an EF1 and PO allocation has long been a major open problem, which we address in this paper.
Moulin~\cite{moulin2019fair} highlighted this issue in his survey (page 436).
Unlike the case of goods, where maximizing Nash social welfare has been successfully used to show the existence of EF1 and PO allocations, no such technique is currently known for chores.
Instead, most existing results on EF1 and PO allocations for chores rely on the concept of competitive equilibrium, similar to the approach used for goods.
For goods, algorithms that compute EF1 and fPO allocations are typically based on the competitive equilibrium framework and use a potential function to guarantee termination.
In these settings, price increases and item reallocations consistently increase the potential, ensuring that the algorithm terminates.

However, a major obstacle in extending this approach to chores is the loss of monotonicity in the potential function.
For chores, changes in prices and item transfers can have inconsistent or even conflicting effects on the potential, making it significantly more difficult to prove termination.

As a result, the existence of EF1 and fPO (or PO) allocations for chores was previously known only for a few restricted cases:
instances with two agents~\cite{aziz2022fair};
instances where each agent's valuation is bivalued~\cite{ebadian2022fairly, garg2022fair};
instances with two types of chores~\cite{Aziz2023};
instances with three agents~\cite{garg2023new};
and instances with three types of valuation functions~\cite{garg2024weighted}.
Although these results all focus on restricted instances, some of them involve quite complicated arguments.
This has made it seem very difficult to prove similar existence results for general additive valuation functions.

\subsection{Our Contributions}

In this work, we consider the problem of fair division of indivisible chores among agents with additive cost functions.\footnote{We define the cost function as the valuation function multiplied by -1.}
Our main contribution is to show the existence of an EF1 and PO allocation for indivisible chores under additive cost functions.
This result resolves an important open question for the case of chores under additive costs.
\begin{theorem}\label{thm:main}
    When each agent has an additive cost function, an EF1 and PO allocation always exists.
\end{theorem}

Building on this existence result, we derive several additional results.
First, we show the existence of an EF1 and fPO allocation for indivisible chores under additive cost functions.
\begin{theorem}\label{thm:fPO}
    When each agent has an additive cost function, an EF1 and fPO allocation always exists.
\end{theorem}
Furthermore, we show that when the number of agents is constant, an EF1 and PO allocation can be computed in polynomial time.
\begin{theorem}\label{thm:const}
    When each agent has an additive cost function and the number of agents is constant, an EF1 and PO allocation can be computed in polynomial time.
\end{theorem}

Finally, we extend all of these results to a more general notion of weighted envy-freeness up to one item (wEF1), which accounts for scenarios where agents have different entitlements or priorities.
Although it is possible to develop the arguments using wEF1 from the start, for the sake of clarity and readability, we choose to discuss it separately.
\begin{theorem}\label{thm:weight}
    The following hold:
     \begin{enumerate}
        \item When each agent has an additive cost function, a wEF1 and PO allocation always exists.\\
        \item When each agent has an additive cost function, a wEF1 and fPO allocation always exists.\\
        \item When each agent has an additive cost function and the number of agents is constant, a wEF1 and PO allocation can be computed in polynomial time.
    \end{enumerate}
\end{theorem}
\paragraph{Our techniques}
We briefly explain the proof techniques used in our main contribution, Theorem~\ref{thm:main}.
Let $N=\{1,2,\ldots ,n\}$ denote the set of agents, and $M=\{1,2,\ldots ,m\}$ the set of chores in the instance.
We begin by slightly perturbing the cost functions and then show the existence of an EF1 and fPO allocation for the perturbed instance.
If the perturbation is sufficiently small, the resulting allocation also satisfies EF1 and PO in the original instance.

By classical results in economic theory~\cite{negishi1960welfare, varian1976}, 
an allocation is fPO if and only if it is an optimal solution 
that minimizes a weighted social cost for some positive weight vector 
$(w_i)_{i\in N} \in \mathbb{R}^n_{>0}$.
This weighted social cost minimization can be formulated as a linear program, whose dual problem has variables $(p_j)_{j\in M}$, which can be interpreted as prices (or rewards) for each chore.

We consider a weight vector $\bw$ on the $(n-1)$-dimensional standard simplex $\Delta^{n-1}$.
We also introduce a small shrinking parameter $\tau > 0$ and define weights $\bw'$ on the slightly shrunken simplex by setting $w'_i := \tau + (1 - \tau n) w_i$ for all $i\in N$.
We then consider the weighted social cost minimization problem under $\bw'$.
By choosing $\tau$ sufficiently small, we can guarantee that for any weight vector $\bw$ and for any optimal solution $\bx$ to the LP under $\bw'$, 
there exists an agent $i$ with $w_i > 0$ who is price envy-free (pEF) with respect to $\bx$ and the unique optimal solution $\bp$ to the dual LP.
Here, we say that agent $i$ is pEF if the total price of agent $i$'s bundle is at most that of any other agent.

Using this property, we apply the KKM (Knaster--Kuratowski--Mazurkiewicz) lemma~\cite{Knaster1929},
which is a variant of a fixed-point theorem, to show that
there exists a weight vector $\bw$ such that for every agent $i$, there is an optimal allocation $\bx$ to the LP under $\bw'$ in which agent $i$ is price envy-free with respect to $\bx$ and the unique optimal solution $\bp$ to the dual LP.
Since we use the weights $\bw'$ on the shrunken simplex, it is guaranteed that all these optimal allocations are fPO.

In general, there may not exist an allocation that is both EF1 and fPO 
among these optimal allocations under the desired weights.
However, for the perturbed instance, we can prove that there exists an optimal allocation under these weights that is both EF1 and fPO. 
This completes the proof of Theorem~\ref{thm:main}.

\paragraph{Novelty and significance}
Our existence proof is particularly interesting because it combines a fixed-point theorem with discrete algorithmic techniques. 
More precisely, we use a fixed-point argument to show the existence of desirable weights, and then apply a discrete algorithm to find an EF1 and fPO allocation among the optimal allocations for these weights.

Classical results highlight the power of fixed-point theorems.
For example, 
Nash~\cite{nash1950equilibrium} proved the existence of a mixed-strategy Nash equilibrium by applying Kakutani's fixed-point theorem;
Arrow and Debreu~\cite{arrow1954existence} used the same theorem to prove the existence of market equilibria;
and Su~\cite{edward1999rental} applied Sperner's lemma to show the existence of envy-free cake-cutting.
In contrast, most existence proofs in discrete fair division have relied solely on discrete algorithms.
Our work makes a significant contribution by showing how fixed-point methods can be powerfully combined with discrete algorithms to address challenging problems in discrete fair division.

The idea of using the KKM lemma to guarantee the existence of a desirable weight vector was first introduced by Igarashi and Meunier~\cite{igarashi2025fair}.
However, they applied the KKM lemma directly without combining it with discrete algorithms.
As a result, although they handled a more complex setting with category constraints and mixed goods and chores, they could only prove the existence of a type of fairness that is weaker than EF1 (and depends on $n$) together with PO.

In contrast, our approach differs from~\cite{igarashi2025fair} in two key ways.
First, we combine the KKM lemma with discrete algorithms.
Second, we adopt the concept of price envy-freeness instead of standard envy-freeness when applying the KKM lemma.
This can be seen as a way to look at both the primal (allocation) and the dual (prices) at the same time.
Such primal-dual perspectives are well known to be powerful in combinatorial optimization, and our results suggest that it is also useful in algorithmic game theory.

Finally, the simplicity of our existence proof is also notable.
As mentioned earlier, some existing results for special cases use complicated arguments, and it seems highly difficult to extend them to general additive cost functions.
We believe that the simplicity of our proof comes from our effective use of fixed-point theorem.
We hope that our approach will inspire further research on fair division in other settings.

\subsection{Related Work}

\paragraph{Fair and efficient allocation for divisible items}
The Fisher market model, which originates from the work of Irving Fisher (see~\cite{brainard2005compute}), has been extensively studied in both economics and computer science. 
This model is renowned for its remarkable fairness and efficiency properties.
Notably, Varian~\cite{varian1974equity} showed that when agents have equal budgets, the resulting equilibrium allocation in a Fisher market is both envy-free and Pareto optimal.
For divisible goods, it is known that such equilibria can be computed in polynomial time under additive valuation functions~\cite{devanur2008market, orlin2010improved, vegh2012strongly}. However, for divisible chores, it remains an open question whether a competitive equilibrium can be computed in polynomial time even under additive valuations. An FPTAS for this problem was developed in~\cite{boodaghians2022polynomial}.

\paragraph{Fair and efficient allocation for indivisible items}
For two agents, Aziz et al.~\cite{aziz2022fair} designed an algorithm to find an EF1 and PO allocation of goods and chores, based on the classical Adjusted Winner algorithm~\cite{brams1996fair}. 
The Fisher market approach has also been used to obtain a fair and efficient allocation for indivisible items.
Barman and Krishnamurthy~\cite{barman2019proximity} showed that a PROP1 and fPO allocation can be computed in polynomial time for indivisible goods.
They also showed that $\mathrm{EF}^1_1$ and fPO allocation can be computed in polynomial time for indivisible goods.
Here, PROP1 (proportionality up to one item) and $\mathrm{EF}^1_1$ are fairness concepts weaker than EF1.
Subsequently, Br{\^a}nzei and Sandomirskiy~\cite{branzei2024algorithms} showed similar results for indivisible chores.
More recently, Garg et al.~\cite{garg2025constant} proved that any fair division instance admits a $2$-EF2 and PO allocation, as well as an $(n-1)$-EF1 and PO allocation for indivisible chores.
Here, $2$-EF2 and $(n-1)$-EF1 are approximate fairness notions that are weaker than EF1.

\paragraph{Fair and efficient allocation with constraints}
There is growing interest in fair and efficient allocation under various practical constraints, extending beyond the unconstrained setting.
Shoshan et al.~\cite{shoshan2023efficient} studied the fair division problem under category constraints, and proposed a polynomial-time algorithm for the case of two agents. Specifically, when each category consists of items that are either all goods or all chores for each agent, their algorithm finds an EF1 and PO allocation. In the more general case, where goods and chores are mixed, the algorithm finds an EF[1,1] (envy-free up to one good and one chore) and PO allocation.
Igarashi and Meunier~\cite{igarashi2025fair} extended this result to general settings with $n$ agents, proving the existence of a Pareto-optimal allocation in which each agent can be made envy-free by reallocating at most $n(n-1)$ items.
For budget constraints, Wu et al.~\cite{wu2021budget} showed that any budget-feasible allocation that maximizes the Nash social welfare achieves a $1/4$-EF1 and PO allocation for goods.
Cookson et al.~\cite{cookson2025constrained} investigated fair division under matroid constraints, showing that maximizing the Nash social welfare yields a $1/2$-EF1 and PO allocation for goods.
For a broader overview of fair division under various constraints, we refer the reader to the survey by Suksompong~\cite{suksompong2021constraints}.
\paragraph{Weighted fair division}
In practical settings, agents often have different entitlements or priorities, which naturally motivates the study of weighted fair division concepts such as weighted EF1 (wEF1).
For goods, Chakraborty et al.~\cite{chakraborty2021weighted} showed that a wEF1 allocation, though it does not necessarily guarantee PO,  can be computed using a weighted picking sequence algorithm.
In contrast, for chores, the existence of wEF1 allocations was established by Wu et al.~\cite{wu2025weighted} through a modification of the weighted picking sequence algorithm.
Several of the results already mentioned for the unweighted setting are also known to extend to the weighted case~\cite{chakraborty2021weighted,wu2025weighted,branzei2024algorithms,garg2024weighted}.
For a more detailed overview, we refer the reader to the survey by Suksompong~\cite{suksompong2025weighted}.

\subsection{Organization}
In Section~\ref{sec:pre}, we introduce the fair division model, the relevant notions of fairness and efficiency, and the necessary graph terminology.
We also present the characterization of fPO and define key notions such as price envy-freeness and non-degenerate instances.
Section~\ref{sec:main} provides the proof of the existence of an EF1 and PO allocation.
Section~\ref{sec:fPO} shows the existence of an EF1 and fPO allocation.
Section~\ref{sec:const} provides the proof of Theorem~\ref{thm:const}, and Section~\ref{sec:weight} provides the proof of Theorem~\ref{thm:weight}.
Finally, Section~\ref{sec:con} discusses the conclusions and directions for future work.
Some proofs are deferred to Appendix~\ref{ap:1}.

\section{Preliminaries}\label{sec:pre}
For positive integer $\ell$, let $[\ell]$ denote $\{1,\ldots, \ell\}$.
\paragraph{The fair division model for indivisible chores}
A {\it fair division instance} $I$ is represented by a tuple $I = (N,M,\{c_i\}_{i\in N})$, where $N=[n]$ denotes a set of $n$ agents, $M=[m]$ denotes a set of $m$ indivisible chores, and $\{c_i\}_{i\in N}$ represents a set of cost functions of each agent $i\in N$.
In this paper, we assume that each cost function $c_i: 2^M \rightarrow \mathbb{R}_{\ge 0}$ is additive, i.e., for any $i \in N$ and any $S \subseteq M$, we have $c_i(S) = \sum_{j \in S} c_i(\{j\})$, with $c_i(\emptyset) = 0$ for all $i \in N$.
To simplify notation, we will write $c_{ij}$ instead of $c_i(\{j\})$ for a singleton chore $j\in M$.
We denote $c_{\max}:=\max_{i \in N, j \in M} c_{ij}$ and $c_{\min}:=\min_{i \in N, j \in M} c_{ij}$.

An {\it allocation} $\bx=(\bx_i)_{i \in N}$ is an $n$-partition of $M$, where $\bx_i\subseteq M$ is the {\it bundle} allocated to agent $i$.
Given an allocation $\bx$, the cost of agent $i\in N$ for the bundle $\bx_i$ is $c_i(\bx_i)=\sum_{j\in \bx_i} c_{ij}$.
A {\it fractional allocation} $\bx = (\bx_i)_{i \in N}$ represents a fractional assignment of the chores to the agents, where each $\bx_i$ is a vector $(x_{ij})_{j \in M}$, and $x_{ij} \in [0,1]$ denotes the fraction of chore $j$ allocated to agent $i$. Additionally, this allocation satisfies the condition that $\sum_{i \in N} x_{ij} = 1$ for each chore $j \in M$.
Given a fractional allocation $\bx$, the cost of agent $i\in N$ for $\bx_i$ is $c_i(\bx_i)=\sum_{j\in M} x_{ij}c_{ij}$.
We will use ``allocation'' to mean an integral allocation and specify ``fractional allocation'' otherwise.

\paragraph{Fairness notions}
Given a fair division instance $I=(N,M,\{c_i\}_{i\in N})$ and an allocation $\bx=(\bx_i)_{i \in N}$, 
we say that an agent $i\in N$ {\it envies} another agent $i'\in N$ if $c_i(\bx_i) > c_i(\bx_{i'})$.
An allocation $\bx$ is said to be {\it envy-free} (EF) if no agent envies any other agent.
An allocation $\bx$ is said to be {\it envy-free up to one chore} (EF1) if for any pair of agents $i,i'\in N$ where $i$ envies $i'$, there exists a chore $j\in \bx_i$ such that $c_i(\bx_i\setminus \{j\})\le  c_i(\bx_{i'})$.

\paragraph{Efficiency notions}
Given a fair division instance $I=(N,M,\{c_i\}_{i\in N})$ and an allocation $\bx=(\bx_i)_{i \in N}$, 
we say that $\bx$ is {\it Pareto dominated} by another allocation $\by$ if $c_i(\by_i)\le c_i(\bx_i)$ for every agent $i\in N$, and $c_{i'}(\by_{i'}) < c_{i'}(\bx_{i'})$ for some agent $i'\in N$.
An allocation $\bx$ is said to be {\it Pareto optimal} (PO) if $\bx$ is not Pareto dominated by any other allocation.
Similarly, a (fractional) allocation $\bx$ is said to be {\it fractionally Pareto optimal} (fPO) if $\bx$ is not Pareto dominated by any fractional allocation.
Note that a fractionally Pareto optimal allocation is also Pareto optimal, but not vice versa.

\paragraph{Characterization of fPO by linear programming}
Classical results in economic theory~\cite{negishi1960welfare, varian1976} establish a connection between fractional Pareto-optimal allocations and weighted utilitarian welfare. For divisible chores with additive cost functions, this result is stated and simply proved in~\cite{branzei2024algorithms}.

\begin{lemma}[Lemma 2.1 in~\cite{branzei2024algorithms}]\label{lem:chaoffPO}
A fractional allocation $\bx$ is fPO if and only if there exists a positive weight vector  $\bw = (w_i)_{i \in N} \in \mathbb{R}_{>0}^n$ such that $\bx$ minimizes the weighted social cost $\sum_{i \in N} w_i c_i(\bx_i)$ over all fractional allocations.
\end{lemma}
Given a positive weight vector $\bw=(w_i)_{i\in N} \in \mathbb{R}^n_{>0}$, 
consider the following primal-dual pair of linear programs:
\begin{center}
\begin{minipage}[t]{0.48\textwidth}
\raggedright
$\mathbf{LP}(\bw)$
\begin{align*}
\text{Min.} \quad & \sum_{i \in N} \sum_{j \in M} w_i  c_{ij} x_{ij} \\
\text{s.t.} \quad & \sum_{i \in N} x_{ij} = 1 \quad \forall j \in M \\
& x_{ij} \geq 0 \quad \forall i \in N,\ \forall j \in M
\end{align*}
\end{minipage}
\hfill
\begin{minipage}[t]{0.48\textwidth}
\raggedright
$\mathbf{Dual\text{-}LP}(\bw)$
\begin{align*}
\text{Max.} \quad & \sum_{j \in M} p_j \\
\text{s.t.} \quad & p_j \leq w_i c_{ij} \quad \forall i \in N,\ \forall j \in M
\end{align*}
\end{minipage}
\end{center}
Each dual variable $p_j$ can be interpreted as the price (or reward) of chore $j$.
Note that the complementary slackness condition is given by $$x_{ij} > 0 \implies p_j = w_i c_{ij} \quad \forall i \in N, \forall j \in M.$$
Since the feasible region of $\mathbf{LP}(\bw)$ represents the set of all fractional allocations, and each integral feasible solution of $\mathbf{LP}(\bw)$ corresponds to an allocation\footnote{If $x_{ij} = 1$, then chore $j$ is allocated to agent $i$.}, Lemma~\ref{lem:chaoffPO} can be restated for integral allocations as follows, providing a characterization of fPO in the discrete setting.

\begin{lemma}\label{lem:cha}
An allocation $\bx$ is fPO if and only if there exists a positive weight vector $\bw = (w_i)_{i \in N} \in \mathbb{R}_{>0}^n$ such that $\bx$ is an optimal allocation to $\mathbf{LP}(\bw)$.
\end{lemma}

\paragraph{Price envy-freeness}
Price envy-freeness (pEF) and price envy-freeness up to one item (pEF1) are key concepts introduced by Barman et al.~\cite{barman2018finding}. They are widely used for finding EF1 and PO allocations for both goods and chores~\cite{barman2018finding, ebadian2022fairly, mahara2024polynomial, garg2022fair, garg2023new, garg2024weighted}.

Let $I = (N, M, \{c_i\}_{i \in N})$ be a fair division instance. Let $\bx = (\bx_i)_{i \in N}$ be an allocation and $\bp = (p_j)_{j \in M}$ a price vector, where each $p_j$ is a positive real number. For any subset of chores $S \subseteq M$, let $\bp(S) := \sum_{j \in S} p_j$ denote the total price of the chores in $S$.
An agent $i \in N$ is said to \emph{price envy} another agent $i' \in N$ if $\bp(\bx_i) > \bp(\bx_{i'})$, and is \emph{price envy-free} if this inequality does not hold for any other agent. 
An allocation $\bx$ is called \emph{price envy-free} (pEF) if every agent is price envy-free. An allocation $\bx$ is called \emph{price envy-free up to one chore} (pEF1) if, for any pair of agents $i, i' \in N$ where $i$ price envies $i'$, there exists a chore $j \in \bx_i$ with $\bp(\bx_i \setminus \{j\}) \le \bp(\bx_{i'})$.

For any $S \subseteq M$, let $\hat{\bp}(S)$ denote the total price of $S$ after removing the most expensive chore. Formally, $\hat{\bp}(S)$ is defined as follows:
\begin{align*}
  \hat{\bp}(S):=
  \left\{
    \begin{array}{ll}
      \min_{j\in S} \bp(S\setminus \{j\}) & {\rm if}~S \neq \emptyset,\\
      0 & {\rm otherwise}.
    \end{array}
  \right.
\end{align*}

\begin{observation}\label{ob:pEF1}
An allocation $\bx$ is pEF1 if and only if $\max_{i \in N} \hat{\bp}(\bx_{i}) \le \min_{i' \in N} \bp(\bx_{i'})$ holds.
\end{observation}
\begin{proof}
From the definition, we directly obtain 
    \begin{align*}
    \text{An allocation}~\bx~\text{is pEF1}
     &\Longleftrightarrow \forall i, i' \in N, \hat{\bp}(\bx_{i}) \le \bp(\bx_{i'})  \\
     &\Longleftrightarrow \max_{i \in N} \hat{\bp}(\bx_{i}) \le \min_{i' \in N} \bp(\bx_{i'}).
    \end{align*}
\end{proof}

Similar to the case of goods~\cite{barman2018finding}, pEF1 implies EF1 for a suitable choice of prices $\bp$.

\begin{lemma}\label{lem:pEF1isEF1}
Given a positive weight vector $\bw \in \mathbb{R}_{>0}^n$, let $\bx$ be an optimal allocation to $\mathbf{LP}(\bw)$, and let $\bp$ be an optimal solution to $\mathbf{Dual\text{-}LP}(\bw)$.  
If the allocation $\bx$ is pEF1, then $\bx$ is EF1.
\end{lemma}

\begin{proof}
Let $i, i' \in N$ be any pair of agents such that $i$ envies $i'$.  
Since $i$ envies $i'$, agent $i$ must have at least one chore.
Since the allocation $\bx$ is pEF1, there exists a chore $j \in \bx_i$ such that $\bp(\bx_i \setminus \{j\}) \le \bp(\bx_{i'}).$
Since $\bx$ and $\bp$ are optimal solutions to the primal and dual problems, they satisfy the complementary slackness conditions by the complementary slackness theorem. Hence, we have
\begin{align*}
w_i c_i(\bx_i \setminus \{j\}) &= \bp(\bx_i \setminus \{j\}) \\
&\le \bp(\bx_{i'}) \\
&\le w_i c_i(\bx_{i'}), 
\end{align*}
where the last inequality follows from the dual feasibility condition.

Therefore, we obtain $c_i(\bx_i \setminus \{j\}) \le c_i(\bx_{i'})$.
Since $i$ and $i'$ are arbitrary, it follows that $\bx$ is EF1.
\end{proof}

\paragraph{Graph terminology}
Throughout this paper, we consider only undirected graphs.

Let $G = (V, E)$ be an undirected graph.  
The graph $G$ is said to be \emph{bipartite} if its vertex set $V$ can be partitioned into two disjoint sets $L$ and $R$ such that every edge joins a vertex in $L$ to a vertex in $R$.  
In this case, we write $G = (L, R; E)$.
A bipartite graph $G = (L, R; E)$ is said to be \emph{complete} if every vertex in $L$ is adjacent to every vertex in $R$.

A \emph{path} is a sequence of distinct vertices $(v_1, v_2, \ldots, v_k)$ such that $(v_i, v_{i+1}) \in E$ for all $i \in [k-1]$.  
A \emph{cycle} is a sequence of vertices $(v_1, v_2, \ldots, v_k, v_{k+1})$ such that 
$(v_i, v_{i+1}) \in E$ for all $i \in [k]$, 
$v_1 = v_{k+1}$, 
$v_1, \ldots, v_k$ are all distinct, 
and $k \ge 3$.
A \emph{forest} is an acyclic graph; that is, a graph containing no cycles.
A \emph{matching} is a set of pairwise disjoint edges.  
A matching is said to \emph{cover} a subset $U \subseteq V$ if every vertex in $U$ is incident to some edge in the matching.

For a subset $S \subseteq V$, let $\Gamma_G(S)$ denote the set of \emph{neighbors} of $S$ in $G$, defined by 
$$
\Gamma_G(S) = \{ v \in V \setminus S \mid (s, v) \in E \text{ for some } s \in S \}.
$$
Given a subset $S \subseteq V$, the \emph{induced subgraph} $G[S]$ is the graph with vertex set $S$ and edge set $\{(u,v) \in E \mid u, v \in S\}$.
For a vertex $v \in V$, let $d_G(v)$ denote the \emph{degree} of $v$.

The following is a classical result in discrete mathematics.
\begin{theorem}[Hall's theorem~\cite{hall1935representatives}]\label{thm:hall}
Let $G=(L, R; E)$ be a bipartite graph. 
Then, $G$ has a matching that covers $L$ if and only if $|S| \le |\Gamma_G(S)|$ for any $S\subseteq L$.
\end{theorem}

\paragraph{Non-degenerate instance}
The notion of non-degenerate instances used in this paper has appeared across various problem settings~\cite{duan2016improved, branzei2024algorithms, bogomolnaia2019dividing}, and it plays a crucial role in our existence proof.

Given a fair division instance $I=(N,M,\{c_i\}_{i\in N})$, consider a cycle $C$ in a complete bipartite graph $(N,M; E)$ given by $C= (i_1,j_1,i_2,j_2,\ldots , i_\ell, j_\ell, i_{\ell+1})$.
We define the product of cost along the cycle as 
$$\pi(C):=\prod_{k=1}^\ell \frac{c_{i_kj_k}}{c_{i_{k+1}j_k}}.$$
We say that a fair division instance $I=(N,M,\{c_i\}_{i\in N})$ is {\it non-degenerate} if the product $\pi(C) \neq 1$ for any cycle $C$ in the complete bipartite graph $(N,M,E)$.

We can make an instance non-degenerate by perturbing each cost value.  
In this paper, we adopt the perturbation technique used in~\cite{duan2016improved}.
Given any instance $I$, we construct a perturbed instance $I^\epsilon = (N, M, \{c^{(\epsilon)}_i\}_{i \in N})$, 
where each perturbed cost is defined as $c^{(\epsilon)}_{ij} := c_{ij} q_{ij}^\epsilon$\footnote{The superscript~$\epsilon$ in~$c^{(\epsilon)}_{ij}$ denotes an index indicating dependence on~$\epsilon$, 
whereas the superscript~$\epsilon$ in~$q_{ij}^{\epsilon}$ denotes exponentiation.} for the original cost $c_{ij}$.  
Here, each $q_{ij}$ denotes the $(m(i - 1) + j)$-th smallest prime number.

We define two positive quantities $\delta$ and $\delta'$ associated with the instance $I$ as follows:

$$
\delta := \min_{i\in N} \min_{S,T: c_i(S)\neq c_i(T)} |c_i(S)-c_i(T)|,
$$

$$
\delta' := \min \left\{
\left| \prod_{k=1}^\ell c_{i_k j_k} - \prod_{k=1}^\ell c_{i_{k+1} j_k} \right|
\;\middle|\;
C = (i_1, j_1, \dots, i_{\ell}, j_{\ell}, i_{\ell+1}) \text{ is a cycle in }  (N, M; E),\ \pi(C) \neq 1
\right\}.
$$

The following three lemmas show that, for sufficiently small $\epsilon > 0$, any EF1 and fPO allocation in the perturbed instance $I^{\epsilon}$ is also EF1 and PO in the original instance $I$.
Proofs of these lemmas are given in Appendix~\ref{ap:1}.
\begin{lemma}\label{lem:non-degenerate}
Let $I = (N, M, \{c_i\}_{i \in N})$ be a fair division instance. 
If 
$0 < \epsilon < \log_{(q_{nm})^n}\left( 1 + \frac{\delta'}{2  (c_{\max})^n} \right)$, 
then the perturbed instance $I^\epsilon$ is non-degenerate.
\end{lemma}

\begin{lemma}\label{lem:ef1}
Let $I = (N, M, \{c_i\}_{i \in N})$ be a fair division instance, and suppose that 
$0 < \epsilon < \log_{q_{nm}}\left( 1 + \frac{\delta}{2m c_{\max}} \right)$. 
If an allocation $\bx$ is EF1 in $I^{\epsilon}$, then $\bx$ is also EF1 in $I$.
\end{lemma}

\begin{lemma}\label{lem:po}
Let $I = (N, M, \{c_i\}_{i \in N})$ be a fair division instance, and suppose that 
$0 < \epsilon < \log_{q_{nm}}\left( 1 + \alpha \right)$, where $\alpha: = \frac{{\delta}^{n}}{2m (q_{nm})^{n-1} (c_{\max})^{n}}$.
If an allocation $\bx$ is fPO in $I^{\epsilon}$, then $\bx$ is PO in $I$.
\end{lemma}

\section{Existence of EF1 and PO Allocation of Indivisible Chores}\label{sec:main}
In this section, we prove Theorem~\ref{thm:main}.
Without loss of generality, we may assume that $c_{ij} > 0$ for all $i\in N$ and $j\in M$. 
Indeed, if $c_{ij} = 0$ for some $i$ and $j$, we can simply assign chore $j$ to agent $i$ and remove it from the instance. 
It is easy to verify that this operation does not affect both EF1 and PO properties.

By Lemmas~\ref{lem:non-degenerate}--\ref{lem:po}, to prove the existence of an EF1 and PO allocation, it suffices to show the existence of an EF1 and fPO allocation in a non-degenerate instance.
Hence, in this section, we assume that the instance is non-degenerate.

Our goal is to find an EF1 and fPO allocation in a non-degenerate instance $I = (N, M, \{c_i\}_{i \in N})$.
By Lemmas~\ref{lem:cha} and~\ref{lem:pEF1isEF1}, it suffices to find a positive weight vector $\bw \in \mathbb{R}_{>0}^n$ such that $\bx$ is an optimal allocation to $\mathbf{LP}(\bw)$ and $\bp$ is an optimal solution to $\mathbf{Dual\text{-}LP}(\bw)$, with the property that $\bx$ is pEF1.
We first determine such a desirable weight vector $\bw$ using a fixed-point argument.

Let $\Delta^{n-1}$ denote the $(n{-}1)$-{\it dimensional standard simplex}, defined as
$\Delta^{n-1} := \{ (x_1, \dots, x_n) \in \mathbb{R}^n \mid \sum_{i=1}^n x_i = 1,\ x_i \ge 0 \ \forall i \in [n] \}$.
For any $\bw \in \Delta^{n-1}$ and {\it shrinking parameter} $\tau$ with $0 < \tau < \frac{c_{\min}}{2n c_{\max}}$, we define the following pair of primal and dual linear programs:
\begin{center}
\small
\begin{minipage}[t]{0.46\textwidth}
\raggedright
$\mathbf{LP}(\bw, \tau)$
\begin{align*}
\text{Min.} \quad & \sum_{i \in N} \sum_{j \in M} (\tau + (1-\tau n)w_i)  c_{ij} x_{ij} \\
\text{s.t.} \quad & \sum_{i \in N} x_{ij} = 1 \quad \forall j \in M \\
& x_{ij} \geq 0 \quad \forall i \in N, \forall j \in M
\end{align*}
\end{minipage}
\hfill
\begin{minipage}[t]{0.46\textwidth}
\raggedright
$\mathbf{Dual\text{-}LP}(\bw, \tau)$
\begin{align*}
\text{Max.} \quad & \sum_{j \in M} p_j \\
\text{s.t.} \quad & p_j \leq (\tau + (1-\tau n)w_i) c_{ij} \quad \forall i \in N, \forall j \in M \\
\end{align*}
\end{minipage}
\end{center}
Note that the complementary slackness condition is given by $$x_{ij} > 0 \implies p_j = (\tau + (1-\tau n)w_i) c_{ij} \quad \forall i \in N, \forall j \in M.$$
\begin{observation}\label{ob:1}
The following properties hold for any $\bw \in \Delta^{n-1}$ and $\tau$ with $0 < \tau < \frac{c_{\min}}{2n c_{\max}}$:
\begin{enumerate}
    \item Define $w'_i := \tau + (1 - \tau n) w_i$ for all $i \in [n]$. Then $\bw' = (w'_i)_{i \in N}$ is a positive vector in $\Delta^{n-1}$.
    
    \item The primal linear program $\mathbf{LP}(\bw, \tau)$ admits an integral optimal solution.
    
    \item Any integral optimal solution of $\mathbf{LP}(\bw, \tau)$ yields an fPO allocation.
    
    \item The dual linear program $\mathbf{Dual\text{-}LP}(\bw, \tau)$ has a unique optimal solution $\bp^* = (p^*_j)_{j \in M}$. Moreover, we have $p^*_j > 0$ for all $j \in M$.
\end{enumerate}
\end{observation}
\begin{proof}
The first claim follows directly by summing the components of $\bw'$.  
The second follows from the fact that the constraint matrix of $\mathbf{LP}(\bw, \tau)$ is totally unimodular.  
The third claim follows directly from Lemma~\ref{lem:cha}.  
Finally, the last claim follows from the structure of $\mathbf{Dual\text{-}LP}(\bw, \tau)$, since the optimal price for each $j \in M$ is given by $p^*_j = \min_{i \in N} (\tau + (1 - \tau n) w_i) c_{ij} > 0$.
\end{proof}

Since the shrinking parameter is sufficiently small, the following lemma holds.

\begin{lemma}\label{lem:tau}
Let $\bw \in \Delta^{n-1}$ be any point in the $(n{-}1)$-dimensional standard simplex.  
Let $\bx$ be any optimal allocation to $\mathbf{LP}(\bw, \tau)$, and let $\bp$ be the optimal solution to $\mathbf{Dual\text{-}LP}(\bw, \tau)$.  
Then there exists an agent $i \in N$ such that $w_i > 0$ and agent $i$ is price envy-free with respect to $\bx$ and $\bp$.
\end{lemma}

\begin{proof}
If $w_i > 0$ for all $i \in N$, the claim follows immediately.  
Now suppose that $w_k = 0$ for some agent $k \in N$.  
Let $h \in N$ be any agent such that $w_h \ge \frac{1}{n}$.
We first show that
$(\tau + (1 - \tau n) w_h )c_{hj} > \tau c_{kj}$ holds for any $j \in M$.  
Rearranging gives $\tau(c_{kj} + n w_h c_{hj} - c_{hj}) < w_h c_{hj}$.
We have
\begin{align*}
\tau(c_{kj} + n w_h c_{hj} - c_{hj}) 
&< \tau(c_{kj} + n w_h c_{hj}) \\
&\le \tau(2n w_h c_{\max}) \\
&< w_h c_{\min} \\
&\le  w_h c_{hj},
\end{align*}
where the second inequality follows from $n w_h \ge 1$, and the third follows from $\tau < \frac{c_{\min}}{2n c_{\max}}$.

Hence, the inequality holds, which implies that agent $h$ receives no chore in $\bx$.  
Otherwise, reallocating any chore from $h$ to $k$ would strictly decrease the objective value of $\mathbf{LP}(\bw, \tau)$, contradicting the optimality of $\bx$.
Therefore, agent $h$ is price envy-free with respect to $\bx$ and $\bp$.
\end{proof}

The following is a classical result known as the KKM (Knaster–Kuratowski–Mazurkiewicz) Lemma~\cite{Knaster1929}, which can be viewed as a set-covering version of Sperner's lemma.

\begin{lemma}[KKM Lemma~\cite{Knaster1929}]\label{lem:kkm}
Let $C_1, C_2, \dots, C_n$ be closed subsets of the $(n{-}1)$-dimensional standard simplex $\Delta^{n-1}$.  
Suppose that for every point $(w_1, w_2, \dots, w_n) \in \Delta^{n-1}$, there exists an index $i \in [n]$ such that $w_i > 0$ and $(w_1, w_2, \dots, w_n) \in C_i$.  
Then the intersection of all the sets is nonempty, i.e.,
$$
\bigcap_{i=1}^{n} C_i \ne \emptyset.
$$
\end{lemma}

Now we apply the KKM lemma to show the existence of a desirable weight vector $\bw \in \Delta^{n-1}$.

\begin{lemma}\label{lem:existw}
There exists a weight vector $\bw \in \Delta^{n-1}$ such that for each agent $i \in N$, there exists an optimal allocation $\bx$ to $\mathbf{LP}(\bw, \tau)$ in which agent $i$ is price envy-free with respect to $\bx$ and the unique optimal solution $\bp$ to $\mathbf{Dual\text{-}LP}(\bw, \tau)$.
\end{lemma}

\begin{proof}
For each $i \in N$, define the subset $C_i \subseteq \Delta^{n-1}$ to be the set of weight vectors $\bw = (w_i)_{i \in N}$ such that there exists an optimal allocation $\bx = (\bx_i)_{i \in N}$ to $\mathbf{LP}(\bw, \tau)$ and the corresponding dual optimal solution $\bp$ for which agent $i$ is price envy-free with respect to $\bx$ and $\bp$.

We first show that each $C_i$ is closed. Let $\{\bw^t\}_{t \in \mathbb{N}} \subseteq C_i$ be a sequence of weight vectors converging to some $\bw^* \in \Delta^{n-1}$. For each $t$, there exists an optimal allocation $\bx^t$ to $\mathbf{LP}(\bw^t, \tau)$ and a corresponding dual solution $\bp^t$ to $\mathbf{Dual\text{-}LP}(\bw^t, \tau)$ such that agent $i$ is price envy-free with respect to $\bx^t$ and $\bp^t$. Since the number of possible allocations is finite, we may assume, by passing to a subsequence, that $\bx^t = \bx^*$ for all $t$. Since the objective function of $\mathbf{LP}(\bw, \tau)$ is continuous in $\bw$, $\bx^*$ is also optimal to $\mathbf{LP}(\bw^*, \tau)$.
Moreover, it is easy to see that agent $i$ remains price envy-free with respect to $\bx^*$ and $\bp^*$, where $\lim_{t\rightarrow \infty} \bp^t = \bp^*$.
Hence, $\bw^* \in C_i$, implying that $C_i$ is closed.

The covering condition required by the KKM lemma for $\{C_i\}_{i \in N}$ follows directly from Lemma~\ref{lem:tau}.

Therefore, by the KKM lemma (Lemma~\ref{lem:kkm}), there exists a weight vector $\bw \in \bigcap_{i \in N} C_i$. By the definition of $C_i$, this means that for each agent $i \in N$, there exists an optimal allocation $\bx$ to $\mathbf{LP}(\bw, \tau)$ in which agent $i$ is price envy-free with respect to $\bx$ and the unique optimal solution $\bp$ to $\mathbf{Dual\text{-}LP}(\bw, \tau)$.
\end{proof}

In what follows, let $\bw \in \Delta^{n-1}$ be a weight vector satisfying the condition stated in Lemma~\ref{lem:existw}, and let $\bp$ denote the unique optimal solution to the dual problem $\mathbf{Dual\text{-}LP}(\bw, \tau)$.  
Define a bipartite graph $G(\bw, \bp) = (N, M; E)$ by
$$
(i, j) \in E \iff p_j = (\tau + (1 - \tau n) w_i) c_{ij}.
$$
We refer to $G(\bw, \bp)$ simply as $G$.
By the complementary slackness theorem, an allocation $\bx$ is optimal to $\mathbf{LP}(\bw, \tau)$ if and only if
$$j \in \bx_i \ \Longrightarrow\ p_j = (\tau + (1 - \tau n) w_i) c_{ij},\ \forall i \in N, \forall j \in M.$$
This is equivalent to
$$
j \in \bx_i \ \Longrightarrow\ (i, j) \in E.
$$
Therefore, the statement of Lemma~\ref{lem:existw} can be equivalently reformulated in terms of the graph $G$:
for every agent $i \in N$, there exists an allocation $\bx$ that assigns chores only along edges in $G$ such that agent $i$ is price envy-free with respect to $\bx$ and $\bp$.
The following lemma shows that under this assumption, one can construct a pEF1 allocation by appropriately assigning chores along the edges of $G$.
The proof of Lemma~\ref{lem:pEF1inG} is deferred to Section~\ref{sec:proof-pEF1inG}.

\begin{lemma}\label{lem:pEF1inG}
There exists an allocation $\bx$ that is both pEF1 and an optimal allocation to $\mathbf{LP}(\bw, \tau)$.
\end{lemma}

Assuming this lemma, we can now prove our main theorem (Theorem~\ref{thm:main}).

\begin{proof}[Proof of Theorem~\ref{thm:main}]
Let $\bx$ be an allocation satisfying the conditions of Lemma~\ref{lem:pEF1inG}.
By the third part of Observation~\ref{ob:1}, we see that $\bx$ is fPO.
Moreover, by Lemma~\ref{lem:pEF1isEF1}, the allocation $\bx$ is also EF1.
Hence, $\bx$ is an EF1 and fPO allocation in the non-degenerate instance.
By Lemmas~\ref{lem:non-degenerate}--\ref{lem:po}, for a sufficiently small perturbation parameter $\epsilon$, 
the allocation $\bx$ remains an EF1 and PO allocation in the original instance.
This completes the proof, showing that an EF1 and PO allocation exists for any instance.
\end{proof}

\subsection{Proof of Lemma~\ref{lem:pEF1inG}}\label{sec:proof-pEF1inG}
It remains to prove Lemma~\ref{lem:pEF1inG}.
We show Lemma~\ref{lem:pEF1inG} by designing a discrete algorithm to find an allocation $\bx$ that is both pEF1 and an optimal allocation to $\mathbf{LP}(\bw, \tau)$.

By the non-degeneracy assumption, the graph $G$ has the following desirable property.

\begin{lemma}
The graph $G = (N, M; E)$ is a forest.
\end{lemma}

\begin{proof}
Suppose, for the sake of contradiction, that $G$ contains a cycle $C = (i_1, j_1, \dots, i_\ell, j_\ell, i_{\ell+1})$.
For convenience, let $j_0 := j_{\ell}$.
Since $(i_k, j_k), (i_k, j_{k-1}) \in E$ for all $k \in [\ell]$, we have
$$p_{j_k} = (\tau + (1 - \tau n) w_{i_k}) c_{i_k j_k}, \quad p_{j_{k-1}} = (\tau + (1 - \tau n) w_{i_k}) c_{i_k j_{k-1}}\quad \forall k \in [\ell].$$ 
Thus, 
$$\frac{c_{i_k j_k}}{c_{i_k j_{k-1}}} = \frac{p_{j_k}}{p_{j_{k-1}}} \quad \forall k \in [\ell].$$ 
Therefore, 
$$\pi(C) = \prod_{k=1}^{\ell} \frac{c_{i_k j_k}}{c_{i_k j_{k-1}}} = \prod_{k=1}^{\ell} \frac{p_{j_k}}{p_{j_{k-1}}} = 1,$$
which contradicts the non-degeneracy of the instance.
Hence, $G = (N, M; E)$ must be a forest.
\end{proof}

\begin{remark}
The fact that $G$ is a forest is crucial for establishing Lemma~\ref{lem:pEF1inG}.  
Indeed, if $G$ contains a cycle, then Lemma~\ref{lem:pEF1inG} may fail to hold, as shown in the following counterexample.
Consider a fair division instance with three agents and four chores.
Suppose that $p_1 = 1$ and $p_2 = p_3 = p_4 = 2$, and let the graph $G$ be as shown in Figure~\ref{fig:1}.
It is easy to verify that for every agent $i\in N$, there exists an allocation that assigns chores only along edges in $G$ such that agent $i$ is price envy-free.
However, there does not exist a pEF1 allocation that assigns chores only along edges in $G$.
\begin{figure}[htbp]
    \centering
    \includegraphics[width=0.35\textwidth]{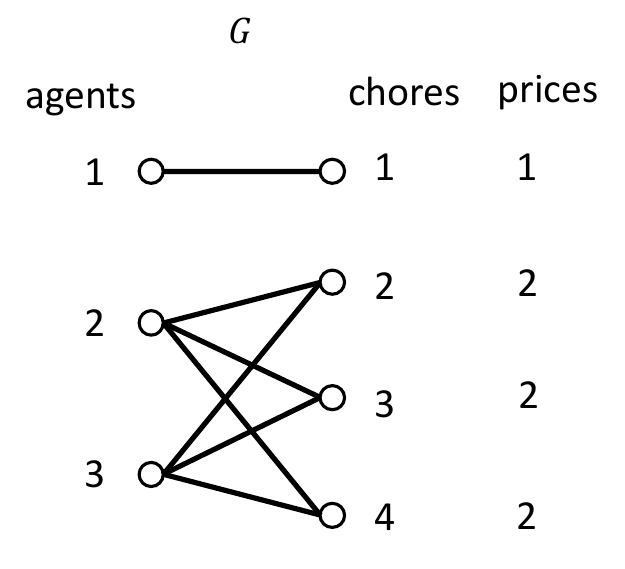}
    \caption{A bipartite graph $G$ containing a cycle that prevents the existence of a pEF1 allocation.}
    \label{fig:1}
\end{figure}

\end{remark}

As a simple observation, any chore with degree one in $G$ must be assigned to a specific agent in every optimal allocation.  
Thus, we first assign all such chores and consider the residual instance.
Let $H = (N, M_H; E_H)$ be the bipartite graph obtained from $G = (N, M; E)$ by removing all chores $j \in M$ with degree one in $G$.
Note that since $G$ is a forest, $H$ is also a forest.

\begin{lemma}\label{lem:n-1}
We have $|M_H| \le n - 1$.
\end{lemma}

\begin{proof}
The graph $H$ has $n + |M_H|$ vertices.  Since $H$ is a forest, it follows that  $|E_H| \le n + |M_H| - 1$.
On the other hand, each vertex $j \in M_H$ has degree at least two in $H$, which implies that $|E_H| \ge 2|M_H|$.
Combining these inequalities yields $2|M_H| \le n + |M_H| - 1$, hence we have $|M_H| \le n - 1$.
\end{proof}

\begin{lemma}\label{lem:covermatching}
There exists a matching in $H$ that covers all vertices in $M_H$.
\end{lemma}

\begin{proof}
By Hall's theorem (Theorem~\ref{thm:hall}), it suffices to show that $|S| \le |\Gamma_H(S)|$ for any $S \subseteq M_H$.
This clearly holds for $S = \emptyset$.

Now fix any nonempty subset $S \subseteq M_H$ and consider the induced subgraph $H[S \cup \Gamma_H(S)]$.
This subgraph has $|S| + |\Gamma_H(S)|$ vertices.
Since $H[S \cup \Gamma_H(S)]$ is a forest, $H[S \cup \Gamma_H(S)]$ has at most $|S| + |\Gamma_H(S)| - 1$ edges.
On the other hand, since each vertex $j \in S$ has degree at least two in $H[S \cup \Gamma_H(S)]$, $H[S \cup \Gamma_H(S)]$ has at least $2|S|$ edges.
Combining these yields $2|S| \le |S| + |\Gamma_H(S)| - 1,$
which gives $|S| \le |\Gamma_H(S)| - 1 < |\Gamma_H(S)|$.
Therefore, there exists a matching that covers all vertices in $M_H$.
\end{proof}

For each agent $i \in N$, define
$$
r_i := \sum_{(i, j) \in E \,:\, d_G(j) = 1} p_j
$$
to be the total price of degree-one chores assigned to agent $i$ in $G$.
Let $r_{\max} := \max_{i \in N} r_i$ denote the maximum such total over all agents, and let $ R := \{i \in N \mid r_i = r_{\max} \}$
denote the set of agents who receive this maximum total price of degree-one chores.

For any optimal allocation $\bx = (\bx_i)_{i \in N}$ to $\mathbf{LP}(\bw, \tau)$, we consider the following two invariants:
\begin{description}
\item[(I1)] Every agent in $R$ receives at most one chore from $M_H$, i.e., $|\bx_i \cap M_H| \le 1\quad \forall i \in R$.
\item[(I2)] Every agent receives a bundle whose total price is at least $r_{\max}$, i.e., $\bp(\bx_i) \ge r_{\max}\quad  \forall i \in N$.
\end{description}

Note that by the definition of $R$, invariant (I2) implies that every agent in $N\setminus R$ must receive at least one chore from $M_H$.
The following lemma guarantees the existence of an optimal allocation satisfying invariants (I1) and (I2).

\begin{lemma}\label{lem:initial}
There exists an optimal allocation $\bx = (\bx_i)_{i \in N}$ to $\mathbf{LP}(\bw, \tau)$ satisfying invariants \textup{(I1)} and \textup{(I2)}.
\end{lemma}

\begin{proof}
Consider the set of chores $\Gamma_H(N \setminus R)$.
We first show that these chores can be allocated to agents in $N \setminus R$ such that every agent receives chores whose total price is at least $r_{\max}$.

Let $i \in R$ be arbitrary. Since $\bw \in \Delta^{n-1}$ is a weight vector satisfying the condition stated in Lemma~\ref{lem:existw}, there exists an optimal allocation $\by$ to $\mathbf{LP}(\bw, \tau)$ such that agent $i$ is price envy-free, i.e., $\bp(\by_i) \le \bp(\by_{i'})$ for all $i' \in N$.
Since $i \in R$, we have $r_{\max} \le \bp(\by_i)$, it follows that $\bp(\by_{i'}) \ge r_{\max}$ for all $i' \in N \setminus R$.
If there exists a chore in $\Gamma_H(N \setminus R)$ that is not allocated to $N \setminus R$ under $\by$, then we can reassign it arbitrarily to an agent in $N \setminus R$ without decreasing the total price of any bundle held by agents in $N \setminus R$. Thus, we can ensure that each agent in $N \setminus R$ receives a bundle from $\Gamma_H(N \setminus R)$ whose total price is at least $r_{\max}$.

Next, consider the remaining chores $M_H \setminus \Gamma_H(N \setminus R)$ (which may be empty). By Lemma~\ref{lem:covermatching}, there exists a matching that covers $M_H \setminus \Gamma_H(N \setminus R)$.
By the definition of $\Gamma_H(N \setminus R)$, each of these chores is matched only to agents in $R$. Thus, we can assign them to agents in $R$ according to a matching.

Combining these assignments yields an optimal allocation $\bx$ to $\mathbf{LP}(\bw, \tau)$ that satisfies both \textup{(I1)} and \textup{(I2)}.
\end{proof}

The allocation $\bx$ obtained in Lemma~\ref{lem:initial} satisfies invariants \textup{(I1)} and \textup{(I2)}, but it may not be pEF1.
We use $\bx$ as an initial allocation and then compute a pEF1 allocation algorithmically.

Given an optimal allocation $\bx=(\bx_i)_{i \in N}$ to $\mathbf{LP}(\bw, \tau)$, we define the sets of \emph{violator agents} and \emph{unmatched agents} as follows:
$$V := \{i \in N \mid \hat{\bp}(\bx_i) > r_{\max} \}, 
\qquad U := \{i \in R \mid \bx_i \cap M_H = \emptyset \}.
$$
We say that a path $P = (i_0, j_1, i_1, j_2, \ldots, i_{\ell-1}, j_\ell, i_\ell)$ in the graph $H$ is an \emph{alternating path} if
$j_k \in \bx_{i_k}\ \text{for all } k \in [\ell].$
We refer to an edge $(i,j)$ in $H$ with $j\in \bx_i$ as an ${\it allocation\ edge}$.

To obtain a pEF1 allocation, the set of violators 
$V$ must be empty.
By removing chores from the bundles of violators, they can be turned into non-violators.
At the same time, assigning chores to unmatched agents does not violate invariant (I1).
Therefore, the algorithm tries to transfer chores from agents in 
$V$ to those in 
$U$, while carefully ensuring that invariants (I1) and (I2) are preserved.

\begin{lemma}\label{lem:pathexist}
Suppose that an optimal allocation $\bx = (\bx_i)_{i \in N}$ to $\mathbf{LP}(\bw, \tau)$ satisfies invariants \textup{(I1)} and \textup{(I2)}, and that $\bx$ is not pEF1. Then the following hold:
\begin{enumerate}
    \item $U \neq \emptyset$.
    \item $\emptyset \neq V \subseteq N \setminus R$.
    \item There exists an alternating path from $U$ to $V$ in $H$.
\end{enumerate}
\end{lemma}

\begin{proof}
For the first statement, observe that (I2) implies that every agent in $N \setminus R$ must have at least one chore from $M_H$ in $\bx$. Since Lemma~\ref{lem:n-1} guarantees $|M_H| \leq n - 1$, there must exist at least one agent in $R$ who receives no chore from $M_H$,  
which implies that $U \neq \emptyset$.

For the second statement, since $\bx$ is not pEF1, it follows from Observation~\ref{ob:pEF1} and (I2) that 
$$
\max_{i \in N} \hat{\bp}(\bx_i) > \min_{i' \in N} \bp(\bx_{i'}) \ge r_{\max},
$$
which implies $V \neq \emptyset$. Moreover, by (I1) and (I2), for any $i \in R$, we have $\hat{\bp}(\bx_i) \le r_{\max} \le \min_{i' \in N} \bp(\bx_{i'})$, which implies that $i \notin V$. Hence, $V \subseteq N \setminus R$.

For the third statement, pick any $i_1 \in V$. By (I2), agent $i_1$ must have at least one chore from $M_H$, hence there exists $j_1 \in \bx_{i_1}$ with $(i_1, j_1) \in E_H$. Since $d_H(j_1) \ge 2$, there exists another agent $i_2 \neq i_1$ with $(i_2, j_1) \in E_H$. If $i_2 \in U$, then we have found an alternating path from $U$ to $V$. Otherwise, $i_2 \in R \setminus U$ or $i_2 \in N \setminus R$. In either case, by the definition of $U$ and (I2), there exists $j_2 \in \bx_{i_2}$ with $(i_2, j_2) \in E_H$.
By continuing inductively in this way, we construct an alternating path. Since the number of vertices is finite and $H$ is a forest, this process must eventually reach some agent in $U$. This completes the proof.
\end{proof}

\paragraph{The FindpEF1 algorithm}
Interestingly, our algorithm builds upon the one proposed for the case of goods in~\cite{mahara2024polynomial}.

The \textit{FindpEF1} algorithm (Algorithm \ref{alg:1}) modifies the current optimal allocation $\bx$ to $\mathbf{LP}(\bw, \tau)$ to a new optimal allocation $\bx'$ to $\mathbf{LP}(\bw, \tau)$, while preserving the invariants (I1) and (I2).
By Lemma~\ref{lem:initial}, such an initial optimal allocation is guaranteed to exist.

The algorithm repeatedly performs the following operation as long as the current allocation is not pEF1.
It first finds an alternating path $P=(i_0, j_1, i_1, j_2,\ldots ,i_{\ell-1}, j_\ell, i_\ell)$ of the {\it shortest} length in the graph $H$, where $i_0\in U$ and $i_\ell \in V$.
By Lemma~\ref{lem:pathexist}, at least one such path is guaranteed to exist.
Next, it identifies an index along the path $P$. Specifically, it selects the largest index $a \in [\ell-1]$ such that $r_{\max} \ge \bp(\bx_{i_a} \cup \{j_{a+1}\} \setminus \{j_a\})$. This means that $i_a$ is the closest agent to $i_\ell$ on the path $P$ that satisfies this condition.
If no such $a$ exists, it is set to $0$.
The algorithm then reallocates the chores along the identified path.
Specifically, agent $i_a$ acquires $j_{a+1}$, agent $i_\ell$ relinquishes $j_{\ell}$, and each intermediate agent $i_k$ acquires $j_{k+1}$ and relinquishes $j_k$.
The allocations for the remaining agents remain unchanged.
Finally, the algorithm returns the updated allocation $\bx'$.
Figure~\ref{fig:2} illustrates this reallocation process along the identified path in $H$.
\begin{algorithm}[t]
\caption{An algorithm to find a pEF1 optimal allocation to $\mathbf{LP}(\bw, \tau)$}
\label{alg:1}
\begin{algorithmic}[1]
\Require The graph $H= (N,M_H; E_H)$, the dual optimal solution $\bp=(p_j)_{j \in M}$ to $\mathbf{Dual\text{-}LP}(\bw, \tau)$, and an optimal allocation $\bx=(\bx_i)_{i \in \mathcal{N}}$ to $\mathbf{LP}(\bw, \tau)$ satisfying (I1) and (I2).
\Ensure A pEF1 optimal allocation $\bx'=(\bx'_i)_{i \in \mathcal{N}}$ to $\mathbf{LP}(\bw, \tau)$.
\While{the current allocation $\bx$ is not pEF1}
\State Let $U$ be the set of unmatched agents and $V$ the set of violator agents for $\bx$.
\State Find a {\it shortest} alternating path $P=(i_0, j_1, i_1, j_2,\ldots ,i_{\ell-1}, j_\ell, i_\ell)$ in $H$, where $i_0\in U$ and $i_\ell \in V$.
\State Let $a \in [\ell-1]$ be the {\it largest} index such that $r_{\max} > \bp(\bx_{i_a} \cup \{j_{a+1}\} \setminus \{j_a\})$.
(If no such $a$ exists, set $a=0$.)
\For{$\forall i \in N$}
\If{$i=i_a$}
\State $\bx_{i_a}\gets \bx_{i_a} \cup \{j_{a+1}\}$
\ElsIf{$i=i_{\ell}$}
\State $\bx_{i_{\ell}}\gets \bx_{i_{\ell}} \setminus \{j_{\ell}\}$
\ElsIf{$i=i_k$ with $a<k<\ell$}
\State $\bx_{i_k}\gets \bx_{i_k} \cup \{j_{k+1}\} \setminus \{j_k\}$
\Else
\State $\bx_i\gets \bx_i$
\EndIf
\EndFor
\EndWhile
\Return $\bx$
\end{algorithmic}
\end{algorithm}
\begin{figure}[tbp]
    \centering
    \includegraphics[width=0.5\textwidth]{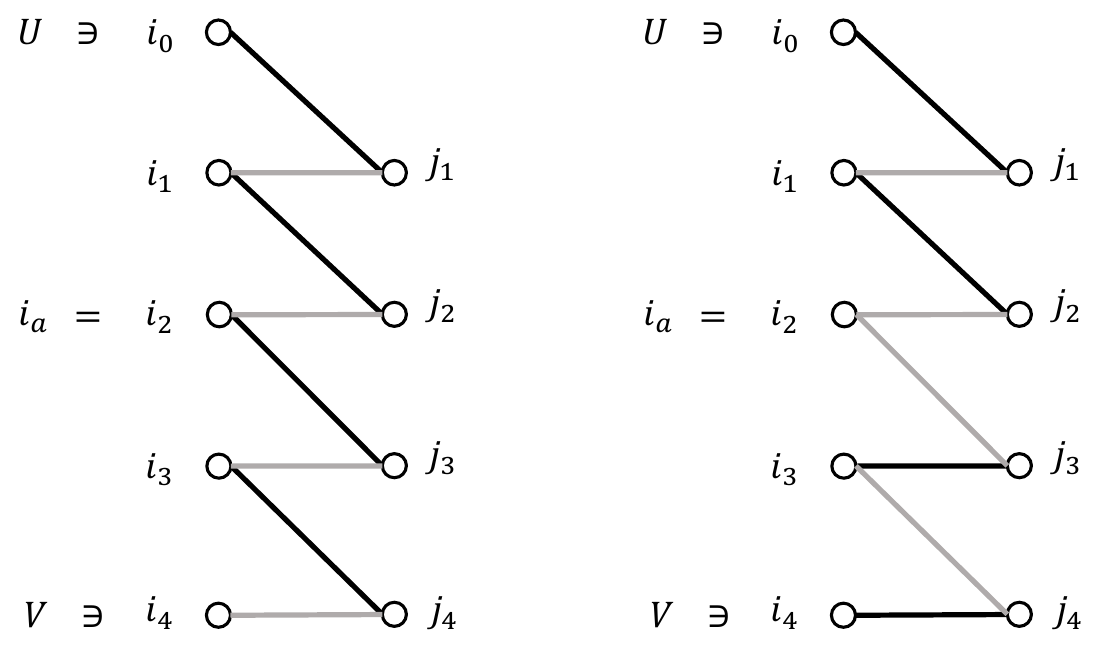}
    \caption{Illustration of the FindpEF1 algorithm (Algorithm~\ref{alg:1}). The left figure shows the state before applying a transfer operation, and the right figure shows the state after applying it. The black edges represent edges in $H$, while the gray edges represent allocation edges.}
    \label{fig:2}
\end{figure}

\paragraph{Correctness and termination of FindpEF1 algorithm}
We now prove the correctness and termination of Algorithm~\ref{alg:1}.
A {\it Transfer step} refers to the item transfer operation performed in lines 5--13 of Algorithm~\ref{alg:1}.
For any positive integer $t$, we define {\it time step $t$} as the state of the algorithm {\it immediately before} the transfer step executed in the $t$-th iteration of Algorithm~\ref{alg:1}.
Let $\bx^t = (\bx^t_i)_{i \in \mathcal{N}}$ denote the current allocation at time step $t$.
Similarly, let $U^t$ denote the set of unmatched agents and $V^t$ denote the set of violator agents at time step $t$.

Note that since item transfers are performed along alternating paths, the algorithm always maintains an optimal allocation to $\mathbf{LP}(\bw, \tau)$.
Hence, if it terminates, Algorithm~\ref{alg:1} returns an allocation that is both pEF1 and optimal for $\mathbf{LP}(\bw, \tau)$.

\begin{lemma}\label{lem:invariant}
Algorithm~\ref{alg:1} maintains invariants \textup{(I1)} and \textup{(I2)} throughout its execution.
\end{lemma}

\begin{proof}
We prove by induction on $t$ that both invariants (I1) and (I2) are maintained in each iteration of the algorithm.
The initial optimal allocation $\bx^1$ satisfies (I1) and (I2) by Lemma~\ref{lem:initial}.

Now assume that at time step $t$, the current allocation $\bx^t$ satisfies \textup{(I1)} and \textup{(I2)}.
Let $P=(i_0, j_1, i_1, j_2,\ldots ,i_{\ell-1}, j_\ell, i_\ell)$ be the shortest alternating path in $H$ computed in Algorithm~\ref{alg:1} during this iteration.
We show that after the transfer step, the updated allocation $\bx^{t+1}$ also satisfies \textup{(I1)} and \textup{(I2)}. 

\smallskip
\noindent\textbf{Invariant (I1):} 
Fix any agent $i \in R$.
We consider the following four cases:
\begin{itemize}
    \item $\bx^{t+1}_i = \bx^t_i$.\\
    In this case, (I1) clearly holds.
    \item $i = i_\ell$.\\
    By Lemma~\ref{lem:pathexist}, we have $i_\ell \in V^t \subseteq N \setminus R$, and hence $i \neq i_\ell$. Thus, this case does not occur.
    \item $i = i_a$.\\
    In this case, agent $i$ must belong to $U^t$. Indeed, if $i \in R \setminus U^t$, then we have $\bp(\bx^t_{i_a} \cup \{j_{a+1}\} \setminus \{j_{a}\})  \ge r_{\max}$ since $i \in R$, contradicting the choice of $i_a$. Thus, (I1) is preserved after the update.
    \item Otherwise.\\
    Agent $i$ acquires one chore from $M_H$ and relinquishes one chore from $M_H$. Hence, the number of chores from $M_H$ assigned to $i$ remains unchanged, and (I1) is preserved.
\end{itemize}

\smallskip
\noindent\textbf{Invariant (I2):} 
Fix any agent $i \in N$.
We consider the following four cases:
\begin{itemize}
    \item $\bx^{t+1}_i = \bx^t_i$.\\
    In this case, (I2) clearly holds.
    \item $i = i_\ell$.\\
    Since $i \in V^t$, we have 
    $$
    \bp(\bx^{t+1}_i) = \bp(\bx^t_i \setminus \{j_\ell\}) \ge \hat{\bp}(\bx^t_i) > r_{\max}.
    $$
    Hence, (I2) continues to hold after the update.
    \item $i = i_a$.\\
    We have
    $$
    \bp(\bx^{t+1}_i) = \bp(\bx^t_i \cup \{j_{a+1}\}) > \bp(\bx^t_i) \ge r_{\max}.
    $$
    Thus, (I2) continues to hold after the update.
    \item Otherwise.\\
    Let $i = i_k$ for some $a < k < \ell$ along the path $P$. 
    Agent $i_k$ acquires $j_{k+1}$ and relinquishes $j_k$. 
    By the choice of index $a$, the total prices of the updated bundle satisfies
    $$
    \bp(\bx^{t+1}_{i_k}) = \bp(\bx^t_{i_k} \cup \{j_{k+1}\} \setminus \{j_k\}) \ge r_{\max},
    $$
    otherwise $i_k$ would have been chosen as $i_a$, contradicting the definition of $a$. 
    Hence, (I2) continues to hold after the update.
\end{itemize}

Therefore, both invariants (I1) and (I2) are maintained in each iteration of the algorithm.
\end{proof}

Let us now discuss the termination of Algorithm~\ref{alg:1}.
For each $i\in N$, we define the {\it level} of agent $i$ at time step $t$, denoted by $\mathrm{level}(i,t)$, as {\it half} the length of the shortest alternating path from $U^t$ to agent $i$ in the graph $H$.
It is easy to verify that, if invariants (I1) and (I2) hold, then for any agent $i$ there exists at least one alternating path from $U^t$ to $i$. This follows by an argument similar to that used in the proof of the third statement in Lemma~\ref{lem:pathexist}.
Note that $\mathrm{level}(i,t) \in \{0,1,\ldots, n-1\}$.

Similarly, we define the level of chore $j \in M_H$ at time step $t$, denoted by $\mathrm{level}(j,t)$, as {\it half} the length of the shortest alternating path from $U^t$ to agent $i$ with $j \in \bx^t_i$ in $H$ such that the path includes $j$.
Note that for any $j\in \bx^t_i$, $\mathrm{level}(j,t) \ge \mathrm{level}(i,t)$ holds.
We say that a chore $j\in M_H$ is {\it critical} for agent $i$ if $j\in \bx^t_i$ and $\mathrm{level}(i,t)=\mathrm{level}(j,t)$.
We denote the set of all critical chores for agent $i$ at time step $t$ by $\overline{\bx}^t_i$.

Finally, we define the potential function $\Phi(t)$ at time step $t$ \footnote{A similar potential function is used in~\cite{barman2018finding}.} as
\begin{equation*}
\Phi(t) :=\sum_{i \in N} n(n-\mathrm{level}(i,t)) + |\overline{\bx}^t_i|.
\end{equation*}

We show that the levels of agents and chores do not decrease during the execution of Algorithm~\ref{alg:1}. 
This monotonicity will be crucial for establishing the termination of the algorithm.
\begin{lemma}\label{lem:levelnondec}
For any time step $t$, we have $\mathrm{level}(i,t) \le \mathrm{level}(i,t+1)$ and $\mathrm{level}(j,t) \le \mathrm{level}(j,t+1)$ for all $i\in N$ and $j \in M_H$.
\end{lemma}
\begin{proof}
    Let $P=(i_0, j_1, i_1, j_2,\ldots ,i_{\ell-1}, j_\ell, i_\ell)$ be a shortest alternating path in $H$, where $i_0 \in U^t$ and $i_\ell \in V^t$, computed by Algorithm~\ref{alg:1} in the $t$-th iteration.
    Let $a\in \{0\} \cup [\ell-1]$ denote the index determined by Algorithm~\ref{alg:1} at the same iteration.
    To simplify the analysis, we assume that instead of transferring multiple chores simultaneously in each iteration, the algorithm transfers the chores $j_\ell, j_{\ell-1}, \ldots, j_{a+1}$ one by one in this order. Moreover, we introduce a finer sequence of time steps $t = t_0 < t_1 < \cdots < t_{\ell - a - 1} = t+1$. Specifically, time step $t_r$ for $r \in \{0,1,\ldots, \ell - a -1 \}$ denotes the state {\it immediately before} the chore $j_{\ell - r}$ is transferred from agent $i_{\ell - r}$ to agent $i_{\ell - r - 1}$.

    We show that $\mathrm{level}(i,t) \le \mathrm{level}(i,t+1)$ for all $i\in N$.
    The argument for $\mathrm{level}(j,t) \le \mathrm{level}(j,t+1)$ for $j \in M_H$ is analogous. 
    It suffices to show that $\mathrm{level}(i,t_r) \le \mathrm{level}(i,t_{r+1})$ for all $i\in N$ for any $r \in \{0,1,\ldots, \ell - a -1 \}$.

    Fix an arbitrary $r \in \{0,1,\ldots, \ell - a -1 \}$ and $i\in N$.
    Observe that the path $(i_0, j_1,\ldots, i_{\ell-r-1}, j_{\ell-r}, i_{\ell-r})$ is a shortest alternating path from $U^{t_r}$ to $i_{\ell-r}$ at time step $t_r$.
    Thus, we have $\mathrm{level}(i_q,t_r) = q$ for all $q \in \{0\}\cup [\ell-r]$.
    Moreover, since agents in $R$ once matched never become unmatched again, we have $U^{t_r}\supseteq U^{t_{r+1}}$.
    We consider the following two cases:
    \begin{itemize}
        \item $i=i_{\ell-r-1}$.\\
        Suppose, for the sake of contradiction, that $\mathrm{level}(i_{\ell-r-1},t_r) > \mathrm{level}(i_{\ell-r-1},t_{r+1})$. 
        Consider a shortest alternating path $Q$ from $U^{t_{r+1}}$ to agent $i_{\ell-r-1}$ at time step $t_{r+1}$.
        This path $Q$ must include the newly created allocation edge $(i_{\ell-r-1},j_{\ell-r})$.
        
        Note that the subpath of $Q$ from its starting point to $j_{\ell-r}$ exists at both time steps $t_{r}$ and $t_{r+1}$.
        If $Q$ contains $i_{\ell-r}$, then we have $\mathrm{level}(i_{\ell-r},t_{r}) < \mathrm{level}(i_{\ell-r-1},t_{r+1})$.
        Otherwise, by concatenating the subpath of $Q$ from its starting point to $j_{\ell-r}$ with the edge $(i_{\ell-r},j_{\ell-r})$ at time step $t_r$, we have $\mathrm{level}(i_{\ell-r},t_{r}) \le \mathrm{level}(i_{\ell-r-1},t_{r+1})$.
        In either case, we obtain 
        $$\ell-r-1=\mathrm{level}(i_{\ell-r-1},t_r) > \mathrm{level}(i_{\ell-r-1},t_{r+1}) \ge \mathrm{level}(i_{\ell-r},t_{r})=\ell-r,$$
        which is a contradiction.
        Therefore, we conclude that $\mathrm{level}(i_{\ell-r-1},t_r) \le \mathrm{level}(i_{\ell-r-1},t_{r+1})$.
        
        \item Otherwise.\\
        Consider a shortest alternating path $Q'$ from $U^{t_{r+1}}$ to agent $i$ at time step $t_{r+1}$.
        If $Q'$ does not contain the newly created allocation edge $(i_{\ell-r-1},j_{\ell-r})$, then $Q'$ exists at both time steps $t_{r}$ and $t_{r+1}$, and we have $\mathrm{level}(i,t_r) \le \mathrm{level}(i,t_{r+1})$.
        Otherwise, if $Q'$ contains $(i_{\ell-r-1},j_{\ell-r})$, it decomposes into a path to $i_{\ell-r-1}$ followed by a path from $i_{\ell-r-1}$ to $i$.
        Since $\mathrm{level}(i_{\ell-r-1},t_r) \le \mathrm{level}(i_{\ell-r-1},t_{r+1})$ and the latter alternating path exists at both time steps $t_{r}$ and $t_{r+1}$, 
        we have $\mathrm{level}(i,t_{r})\le \mathrm{level}(i,t_{r+1})$.
    \end{itemize}

\end{proof}

Using the monotonicity of the levels established above, we can now bound the number of iterations of Algorithm~\ref{alg:1}.
\begin{lemma}\label{lem:terminate}
Algorithm~\ref{alg:1} terminates after $O(n^3)$ iterations.
\end{lemma}
\begin{proof}
    We will show that the value of the potential function decreases by at least one in each iteration of Algorithm~\ref{alg:1}, i.e., $\Phi(t) > \Phi(t+1)$ for any $t$.
    Note that $\Phi(t)$ is non-negative, integral, and its maximum value is $O(n^3)$.
    Thus, once we show this, the proof is complete.

    Fix any time step $t$.
    Let $P=(i_0, j_1, i_1, j_2,\ldots ,i_{\ell-1}, j_\ell, i_\ell)$ be a shortest alternating path in $H$, computed by Algorithm~\ref{alg:1} in the $t$-th iteration.
    Let $a\in \{0\} \cup [\ell-1]$ denote the index determined by Algorithm~\ref{alg:1} at the same iteration.
    Note that $\mathrm{level}(i_k,t)=k$ for any $k\in \{0\}\cup [\ell]$ and $\mathrm{level}(j_k,t)=k$ for any $k\in [\ell]$.

    Fix an arbitrary $i \in N$.
    We analyze the contribution of each agent to $\Phi(t)$ by considering the following two cases:
    \begin{itemize}    
        
        
        \item $i = i_{k}$ for some $a < k \le \ell$.\\
        We show that either $\mathrm{level}(i,t+1) > \mathrm{level}(i,t)$, or $\mathrm{level}(i,t+1)=\mathrm{level}(i,t)$ and $\overline{\bx}^{t+1}_i \subsetneq \overline{\bx}^{t}_i$.
        Note that if this holds, then the contribution of agent $i$ to $\Phi(t)$ strictly decreases, since $|\overline{\bx}^{t}_i| \le n-1$ by Lemma~\ref{lem:n-1}.
        
        By Lemma~\ref{lem:levelnondec}, we know that $\mathrm{level}(i,t+1) \ge \mathrm{level}(i,t)$.
        Suppose that $\mathrm{level}(i,t+1) = \mathrm{level}(i,t)$.
        Now agent $i$ relinquishes the critical chore $j_{k}$.
        Moreover, for any $j \in \bx^t_{i} \setminus \overline{\bx}^t_{i}$, 
        we have $$\mathrm{level}(j,t+1)\ge \mathrm{level}(j,t)> \mathrm{level}(i,t)=\mathrm{level}(i,t+1),$$
        where the first inequality follows from Lemma~\ref{lem:levelnondec} and the second follows from the fact that $j$ was not critical for agent $i$ at time step $t$.
        This shows that any chore that was non-critical at time $t$ remains non-critical at time $t+1$.

        Finally, if $i = i_k$ for some $a < k < \ell$, then agent $i$ acquires the chore $j_{k+1}$.
        We have 
        $$\mathrm{level}(j_{k+1},t+1)\ge \mathrm{level}(j_{k+1},t)=k+1>\mathrm{level}(i,t+1) = \mathrm{level}(i,t) = k.$$
        Thus, $j_{k+1}$ is not critical for agent $i$ at time step $t+1$. 
        Therefore, we conclude that $\overline{\bx}^{t+1}_i \subsetneq \overline{\bx}^{t}_i$.
        
        \item Otherwise.\\
        In this case, either $i = i_a$ or $\bx^{t+1}_i = \bx^{t}_i$.
        We show that either $\mathrm{level}(i,t+1) > \mathrm{level}(i,t)$, or $\mathrm{level}(i,t+1)=\mathrm{level}(i,t)$ and $\overline{\bx}^{t+1}_i \subseteq \overline{\bx}^{t}_i$.
        Note that if this holds, then the contribution of agent $i$ to $\Phi(t)$ does not increase.
        By Lemma~\ref{lem:levelnondec}, we know that $\mathrm{level}(i,t+1) \ge \mathrm{level}(i,t)$.
        Suppose that $\mathrm{level}(i,t+1)=\mathrm{level}(i,t)$.
        Then, by a similar argument as above, we see that any chore that was non-critical at time $t$ remains non-critical at time $t+1$.

        If $i = i_a$, then agent $i$ acquires the chore $j_{a+1}$.
        By a similar argument, we can show that $j_{a+1}$ is not critical for agent $i$ at time step $t+1$.
        Therefore, we conclude that $\overline{\bx}^{t+1}_{i} \subseteq \overline{\bx}^{t}_{i}$.
    \end{itemize}
    Consequently, the contribution of each agent to the potential $\Phi$ does not increase and the contribution of at least one agent to $\Phi$ strictly decreases.
    Therefore, we conclude that $\Phi(t) > \Phi(t+1)$.

\end{proof}

We are now ready to prove Lemma~\ref{lem:pEF1inG}.
\begin{proof}[Proof of Lemma~\ref{lem:pEF1inG}]
    By Lemma~\ref{lem:initial}, it is guaranteed that there exists an optimal allocation to $\mathbf{LP}(\bw, \tau)$ satisfying (I1) and (I2).
    By Lemma~\ref{lem:terminate}, Algorithm~\ref{alg:1} terminates after $O(n^3)$ iterations and outputs a pEF1 optimal allocation to $\mathbf{LP}(\bw, \tau)$.
    This completes the proof.
\end{proof}
\section{Existence of EF1 and fPO Allocation of Indivisible Chores}\label{sec:fPO}
In this section, we prove Theorem~\ref{thm:fPO}.
Note that the proof method closely follows that of Theorem 3.2 in~\cite{barman2018finding} for the setting of goods.

\begin{proof}[Proof of Theorem~\ref{thm:fPO}]
Let $I = (N, M, \{c_i\}_{i \in N})$ be a fair division instance. For each $z \in \mathbb{N}$, define a perturbation parameter $\epsilon_z := \frac{1}{z}\log_{(q_{nm})^n} (1 + \beta)$, where $\beta := \min \{ \frac{\delta'}{2 (c_{\max})^n}, \frac{\delta}{2m c_{\max}}, \frac{{\delta}^{n}}{2m (q_{nm})^{n-1} (c_{\max})^{n}}\}$. 
Let $I^{\epsilon_z}$ denote the instance obtained by perturbing $I$ with $\epsilon_z$.
Note that $\epsilon_z$ is sufficiently small to satisfy the assumptions of Lemmas~\ref{lem:non-degenerate}--\ref{lem:po}.

By the proof of Theorem~\ref{thm:main}, each perturbed instance $I^{\epsilon_z}$ admits an allocation $\bx^z$ that is both EF1 and fPO. Furthermore, by Lemma~\ref{lem:ef1}, each $\bx^z$ is also EF1 in the original instance $I$.

Hence, it suffices to show that the sequence of EF1 allocations $(\bx^z)_{z \in \mathbb{N}}$ contains an fPO allocation in $I$. Since $\bx^z$ is fPO in $I^{\epsilon_z}$, Lemma~\ref{lem:cha} implies that for each $z \in \mathbb{N}$, there exists a positive weight vector $\bw^z = (w^z_i)_{i \in N} \in \mathbb{R}^n_{>0}$ such that $\bx^z$ is an optimal allocation to $\mathbf{LP}(\bw^z)$. By a similar argument to the proof of Lemma~\ref{lem:po}, we may assume that
\begin{equation}\label{eq:a}
\max_{i \in N} w^z_i = 1 
\quad \text{and} \quad
w^z_i \ge \left( \frac{\delta}{q_{nm} c_{\max}} \right)^{n-1}
\quad \text{for all } i \in N.
\end{equation}
Thus, the sequence $(\bw^z)_{z\in \mathbb{N}}$ is bounded.
Since the number of possible allocations is finite, there exists a subsequence $(\bx^{z_k})_{k \in \mathbb{N}}$ such that $\bx^{z_k} = \bx^*$ for all $k$, for some fixed allocation $\bx^*$. By the Bolzano–Weierstrass theorem, there exists a further subsequence $(\bw^{z_{t_k}})_{k \in \mathbb{N}}$ that converges to some $\bw^* \in \mathbb{R}^n_{> 0}$, i.e., $\lim_{k \rightarrow \infty} \bw^{z_{t_k}} = \bw^*$. Reindexing if necessary, we denote this subsequence as $(\bx^{z_k}, \bw^{z_k})_{k \in \mathbb{N}}$, where $\bx^{z_k} = \bx^*$ for any $k \in \mathbb{N}$ and $\lim_{k \rightarrow \infty} \bw^{z_{k}} = \bw^*$.
Note that $w^*_i > 0$ for all $i \in N$ by (\ref{eq:a}).

By construction of the perturbed instance $I^{\epsilon_{z_k}}$, we have
$$
c_{ij} \le c^{\epsilon_{z_k}}_{ij} \le c_{ij} q^{\epsilon_{z_k}}_{ij}, \quad \forall i \in N, \forall j \in M.
$$
Hence, by the sandwich theorem, $\lim_{k \to \infty} c^{\epsilon_{z_k}}_{ij} = c_{ij}$.
Let $\bp^{z_k}$ denote an optimal solution to the dual program $\mathbf{Dual\text{-}LP}(\bw^{z_k})$. From the complementary slackness theorem, for any $j \in \bx^*_i$,
$$
p^{z_k}_j = w^{z_k}_i c_{ij} q^{\epsilon_{z_k}}_{ij}.
$$
Taking the limit as $k \to \infty$, we obtain
$$
\lim_{k \to \infty} p^{z_k}_j = w^*_i c_{ij}, \quad \text{for\ any } j \in \bx^*_i.
$$
Define $\bp^* := \lim_{k \to \infty} \bp^{z_k}$.
From the feasibility of each $\bp^{z_k}$ in the dual program, we have
$$
p^{z_k}_j \le w^{z_k}_i c_{ij} q^{\epsilon_{z_k}}_{ij}, \quad \forall i \in N, \forall j \in M.
$$
Taking the limit, it follows that
$$
p^*_j \le w^*_i c_{ij}, \quad \forall i \in N, \forall j \in M.
$$
Thus, the pair $(\bx^*, \bp^*)$ satisfies the feasibility and complementary slackness conditions for $\mathbf{LP}(\bw^*)$ and $\mathbf{Dual\text{-}LP}(\bw^*)$. 
By the complementary slackness theorem, $\bx^*$ is an optimal solution to $\mathbf{LP}(\bw^*)$, and hence, $\bx^*$ is an fPO allocation in the original instance $I$ by Lemma~\ref{lem:cha}.
\end{proof}
\section{Finding an EF1 and PO Allocation for Constant Number of Agents}\label{sec:const}
In this section, we prove Theorem~\ref{thm:const}.
We assume that each cost value is a positive integer, i.e., $c_{ij}\in \mathbb{Z}_{> 0}$ for all $i\in N$ and $j \in M$.
Our approach combines Theorem~\ref{thm:fPO} with the results from~\cite{branzei2024algorithms}.

We begin by introducing some terminology from~\cite{branzei2024algorithms}.
Given a fractional allocation $\bx$, the {\it consumption~graph} $G_{\bx}$ is defined as the bipartite graph $(N,M;E)$, where $(i,j)\in E$ if and only if $x_{ij} > 0$.
We say that $\bc \in \mathbb{R}^n_{\ge 0}$ is an {\it fPO cost vector} if there exists an fPO fractional allocation $\bx$ such that $\bc = (c_1(\bx_1),\ldots , c_n(\bx_n))$.
A collection of bipartite graphs $\mathcal{G}$ is said to be {\it rich} for a given fair division instance $(N,M,\{c_{i}\}_{i \in N})$ if, for any fPO cost vector $\bc$, there exists a fractional allocation $\bx$ with $\bc=(c_1(\bx_1),\ldots , c_n(\bx_n))$ such that the consumption graph $G_{\bx}$ belongs to $\mathcal{G}$.

Bogomolnaia~\cite{bogomolnaia2019dividing} showed that for each fPO cost vector $\bc$ of a non-degenerate instance, there is a unique fractional allocation $\bx$ such that $\bc=(c_1(\bx_1),\ldots , c_n(\bx_n))$ (see Lemma 2 in~\cite{bogomolnaia2019dividing}).

The following theorem shows that such a rich family of graphs can be constructed in polynomial time when the number of agents is constant.
\begin{theorem}[Lemma 4.4 in~\cite{branzei2024algorithms}]\label{thm:branzei}
For a constant number of agents $n$, a rich family of graphs $\mathcal{G}$ can be constructed in time $O(m^{\frac{n(n-1)}{2}+1})$
and has at most $(2m-1)^{\frac{n(n-1)}{2}}$ elements.
    
\end{theorem}
We are now ready to prove Theorem~\ref{thm:const}.
\begin{proof}[Proof of Theorem~\ref{thm:const}]
We first compute the $nm$ smallest prime numbers $\{q_{ij}\}_{i\in N j \in M}$.
Since it is known that $q_{nm} = O(nm \log nm )$ by the prime
number theorem, we can find these primes in polynomial time using the sieve of Eratosthenes.
Next, we set $\epsilon:= \lfloor \log_{(q_{nm})^n} (1+\frac{1}{2m (q_{nm})^{n-1} (c_{\max})^{n}}) \rfloor$.
Let $I^{\epsilon}$ denote the instance obtained by perturbing $I$ with $\epsilon$.
Since each cost value is a positive integer, we have $\delta \ge 1$ and $\delta' \ge 1$.
Thus, $\epsilon$ is sufficiently small to satisfy the assumptions of Lemmas~\ref{lem:non-degenerate}--\ref{lem:po}.  
By Lemma~\ref{lem:non-degenerate}, $I^\epsilon = (N, M, \{c^{(\epsilon)}_i\}_{i \in N})$ is non-degenerate.

By Theorem~\ref{thm:fPO}, we know that $I^{\epsilon}$ has an EF1 and fPO allocation.
We now show how to compute an EF1 and fPO allocation $\bx$ in $I^\epsilon$ in polynomial time when $n$ is constant.
Then, by Lemmas~\ref{lem:ef1} and~\ref{lem:po}, this allocation $\bx$ will also be EF1 and PO in the original instance $I$.

Since $n$ is constant, we can use Theorem~\ref{thm:branzei} to generate a rich family of bipartite graphs $\mathcal{G}$ in polynomial time.
Furthermore, since $I^{\epsilon}$ is non-degenerate, $\mathbf{x}$ is the unique allocation corresponding to the cost vector $(c_1(\bx_1),\ldots , c_n(\bx_n))$.
Thus, the consumption graph $G_{\mathbf{x}}$ of $\mathbf{x}$ must belong to $\mathcal{G}$.

Therfore, we can go through each graph $G$ in $\mathcal{G}$, check if every chore has degree one, which means it gives an integral allocation, and then check if the allocation is EF1.
Since $\mathcal{G}$ has only a polynomial number of graphs in $m$, and checking EF1 also takes polynomial time, the whole algorithm runs in polynomial time.
\end{proof}

\section{Weighted Fair Division}\label{sec:weight}
In this section, we prove Theorem~\ref{thm:weight}.
Since most of the proof closely follows the proofs of Theorems~\ref{thm:main}--\ref{thm:const}, we only highlight the parts that differ.

We begin by introducing some terminology for weighted fair division.
A {\it weighted fair division instance} $I$ is represented by a tuple $I = (N, M, \{c_i \}_{i\in N}, \{\alpha_i\}_{i\in N})$, where $\alpha_i > 0$ denotes the weight associated with agent $i$. This generalizes a standard fair division instance by assigning a positive weight to each agent.
We denote $\alpha_{\min}:= \min_{i \in N} \alpha_i$.

Given a weighted fair division instance $I = (N, M, \{c_i \}_{i\in N}, \{\alpha_i\}_{i\in N})$ and an allocation $\bx=(\bx_i)_{i \in N}$, 
we say that an agent $i\in N$  {\it weighted envies} another agent $i'\in N$ if $\frac{c_i(\bx_i)}{\alpha_i} > \frac{c_i(\bx_{i'})}{\alpha_{i'}}$.
An allocation $\bx$ is said to be {\it weighted envy-free} (wEF) if no agent weighted envies any other agent.
An allocation $\bx$ is said to be  {\it weighted envy-free up to one chore} (wEF1) if for any pair of agents $i,i'\in N$ where $i$ weighted envies $i'$, there exists a chore $j\in \bx_i$ such that $\frac{c_i(\bx_i\setminus \{j\})}{\alpha_i}\le  \frac{c_i(\bx_{i'})}{\alpha_{i'}}$.

Let $\bp = (p_j)_{j \in M}$ be a price vector.
An agent $i \in N$ is said to \emph{weighted price envy} another agent $i' \in N$ if $\frac{\bp(\bx_i)}{\alpha_i} > \frac{\bp(\bx_{i'})}{\alpha_{i'}}$, and is \emph{weighted price envy-free} if this inequality does not hold for any other agent. 
An allocation $\bx$ is called \emph{weighted price envy-free} (wpEF) if every agent is weighted price envy-free. An allocation $\bx$ is called \emph{weighted price envy-free up to one chore} (wpEF1) if, for any pair of agents $i, i' \in N$ where $i$ weighted price envies $i'$, there exists a chore $j \in \bx_i$ with $\frac{\bp(\bx_i \setminus \{j\})}{\alpha_i} \le \frac{\bp(\bx_{i'})}{\alpha_{i'}}$.

We define the positive quantity $\delta''$ associated with the weighted fair division instance $I$ as follows:

$$
\delta'' := \min_{i, i'\in N} \min_{S,T: \frac{c_i(S)}{\alpha_i}\neq \frac{c_i(T)}{\alpha_{i'}}} \left| \frac{c_i(S)}{\alpha_i}-\frac{c_i(T)}{\alpha_{i'}} \right|.
$$

Similar to the unweighted setting, we have the following properties.
\begin{observation}\label{ob:wpEF1}
An allocation $\bx$ is wpEF1 if and only if $\max_{i \in N} \frac{\hat{\bp}(\bx_{i})}{\alpha_i} \le \min_{i' \in N} \frac{\bp(\bx_{i'})}{\alpha_{i'}}$ holds.
\end{observation}

\begin{lemma}\label{lem:wpEF1iswEF1}
Given a positive weight vector $\bw \in \mathbb{R}_{>0}^n$, let $\bx$ be an optimal allocation to $\mathbf{LP}(\bw)$, and let $\bp$ be an optimal solution to $\mathbf{Dual\text{-}LP}(\bw)$.  
If the allocation $\bx$ is wpEF1, then $\bx$ is wEF1.
\end{lemma}

\begin{lemma}\label{lem:wef1}
Let $I = (N, M, \{c_i\}_{i \in N})$ be a fair division instance, and suppose that 
$0 < \epsilon < \log_{q_{nm}}\left( 1 + \frac{\delta'' \alpha_{\min}}{2m c_{\max}} \right)$. 
If an allocation $\bx$ is wEF1 in $I^{\epsilon}$, then $\bx$ is also wEF1 in $I$.
\end{lemma}

Since the proofs are almost identical to those in the unweighted setting, we omit the proofs of the following two lemmas.

\begin{lemma}\label{lem:wtau}
Let $\bw \in \Delta^{n-1}$ be any point in the $(n{-}1)$-dimensional standard simplex.  
Let $\bx$ be any optimal allocation to $\mathbf{LP}(\bw, \tau)$, and let $\bp$ be the optimal solution to $\mathbf{Dual\text{-}LP}(\bw, \tau)$.  
Then there exists an agent $i \in N$ such that $w_i > 0$ and agent $i$ is weighted price envy-free with respect to $\bx$ and $\bp$.
\end{lemma}

\begin{lemma}\label{lem:wexistw}
There exists a weight vector $\bw \in \Delta^{n-1}$ such that for each agent $i \in N$, there exists an optimal allocation $\bx$ to $\mathbf{LP}(\bw, \tau)$ in which agent $i$ is weighted price envy-free with respect to $\bx$ and the unique optimal solution $\bp$ to $\mathbf{Dual\text{-}LP}(\bw, \tau)$.
\end{lemma}

Finally, similar to the unweighted setting, we have the following result.
\begin{lemma}\label{lem:wpEF1inG}
There exists an allocation $\bx$ that is both weighted pEF1 and an optimal allocation to $\mathbf{LP}(\bw, \tau)$.
\end{lemma}

Assuming this lemma, we can now prove Theorem~\ref{thm:weight}.

\begin{proof}[Proof of Theorem~\ref{thm:weight}]
    For the first statement, 
    let $\bx$ be an allocation satisfying the conditions of Lemma~\ref{lem:wpEF1inG}.
    By the first part of Observation~\ref{ob:1} together with Lemma~\ref{lem:cha}, we see that $\bx$ is fPO.
    Moreover, by Lemma~\ref{lem:wpEF1iswEF1}, the allocation $\bx$ is also wEF1.
    Hence, $\bx$ is a wEF1 and fPO allocation in the non-degenerate instance.
    By Lemmas~\ref{lem:non-degenerate}, \ref{lem:wef1}, and~\ref{lem:po}, for a sufficiently small perturbation parameter $\epsilon$, 
    the allocation $\bx$ remains a wEF1 and PO allocation in the original instance.
    This completes the proof, showing that a wEF1 and PO allocation exists for any instance.

    The second and third claims are omitted since their proofs are almost identical to those of Theorems~\ref{thm:fPO} and~\ref{thm:const}.
\end{proof}

\subsection{Proof of Lemma~\ref{lem:wpEF1inG}}\label{sec:proof-wpEF1inG}
Let $\bw \in \Delta^{n-1}$ be a weight vector satisfying the condition stated in Lemma~\ref{lem:wexistw}, and let $\bp$ denote the unique optimal solution to the dual problem $\mathbf{Dual\text{-}LP}(\bw, \tau)$.  
Define a bipartite graph $G(\bw, \bp) = (N, M; E)$ in the same way as in Section~\ref{sec:main}, and we refer to $G(\bw, \bp)$ simply as $G$.

For each agent $i \in N$, define
$$
r_i := \sum_{(i, j) \in E \,:\, d_G(j) = 1} \frac{p_j}{\alpha_i}
$$
to be the total price of degree-one chores assigned to agent $i$ in $G$, normalized by $\alpha_i$.
Let $r_{\max} := \max_{i \in N} r_i$ denote the maximum such total over all agents, and let $ R := \{i \in N \mid r_i = r_{\max} \}$
denote the set of agents who receive this maximum total of degree-one chores.

For any optimal allocation $\bx = (\bx_i)_{i \in N}$ to $\mathbf{LP}(\bw, \tau)$, we consider the following two invariants:
\begin{description}
\item[(I'1)] Every agent in $R$ receives at most one chore from $M_H$, i.e., $|\bx_i \cap M_H| \le 1\quad \forall i \in R$.
\item[(I'2)] Every agent receives a bundle whose total price normalized by $\alpha_i$ is at least $r_{\max}$, i.e., $\frac{\bp(\bx_i)}{\alpha_i} \ge r_{\max}\quad  \forall i \in N$.
\end{description}

\begin{lemma}\label{lem:winitial}
There exists an optimal allocation $\bx = (\bx_i)_{i \in N}$ to $\mathbf{LP}(\bw, \tau)$ satisfying invariants \textup{(I'1)} and \textup{(I'2)}.
\end{lemma}

Given an optimal allocation $\bx=(\bx_i)_{i \in N}$ to $\mathbf{LP}(\bw, \tau)$, we define the sets of \emph{violator agents} and \emph{unmatched agents} as follows:
$$V := \{i \in N \mid \frac{\hat{\bp}(\bx_i)}{\alpha_i} > r_{\max} \}, 
\qquad U := \{i \in R \mid \bx_i \cap M_H = \emptyset \}.
$$
\begin{lemma}\label{lem:wpathexist}
Suppose that an optimal allocation $\bx = (\bx_i)_{i \in N}$ to $\mathbf{LP}(\bw, \tau)$ satisfies invariants \textup{(I'1)} and \textup{(I'2)}, and that $\bx$ is not wpEF1. Then the following hold:
\begin{enumerate}
    \item $U \neq \emptyset$.
    \item $\emptyset \neq V \subseteq N \setminus R$.
    \item There exists an alternating path from $U$ to $V$ in $H$.
\end{enumerate}
\end{lemma}

\begin{algorithm}[t]
\caption{An algorithm to find a wpEF1 optimal allocation to $\mathbf{LP}(\bw, \tau)$}
\label{alg:2}
\begin{algorithmic}[1]
\Require The graph $H= (N,M_H; E_H)$, the dual optimal solution $\bp=(p_j)_{j \in M}$ to $\mathbf{Dual\text{-}LP}(\bw, \tau)$, and an optimal allocation $\bx=(\bx_i)_{i \in \mathcal{N}}$ to $\mathbf{LP}(\bw, \tau)$ satisfying (I'1) and (I'2).
\Ensure A pEF1 optimal allocation $\bx'=(\bx'_i)_{i \in \mathcal{N}}$ to $\mathbf{LP}(\bw, \tau)$.
\While{the current allocation $\bx$ is not wpEF1}
\State Let $U$ be the set of unmatched agents and $V$ the set of violator agents for $\bx$.
\State Find a {\it shortest} alternating path $P=(i_0, j_1, i_1, j_2,\ldots ,i_{\ell-1}, j_\ell, i_\ell)$ in $H$, where $i_0\in U$ and $i_\ell \in V$.
\State Let $a \in [\ell-1]$ be the {\it largest} index such that $r_{\max} > \frac{\bp(\bx_{i_a} \cup \{j_{a+1}\} \setminus \{j_a\})}{\alpha_{i_a}}$.
(If no such $a$ exists, set $a=0$.)
\For{$\forall i \in N$}
\If{$i=i_a$}
\State $\bx_{i_a}\gets \bx_{i_a} \cup \{j_{a+1}\}$
\ElsIf{$i=i_{\ell}$}
\State $\bx_{i_{\ell}}\gets \bx_{i_{\ell}} \setminus \{j_{\ell}\}$
\ElsIf{$i=i_k$ with $a<k<\ell$}
\State $\bx_{i_k}\gets \bx_{i_k} \cup \{j_{k+1}\} \setminus \{j_k\}$
\Else
\State $\bx_i\gets \bx_i$
\EndIf
\EndFor
\EndWhile
\Return $\bx$
\end{algorithmic}
\end{algorithm}

\begin{lemma}\label{lem:winvariant}
Algorithm~\ref{alg:2} maintains invariants \textup{(I'1)} and \textup{(I'2)} throughout its execution.
\end{lemma}

Since the proof is identical to that in the unweighted setting, we omit the proof of the following lemma.
\begin{lemma}\label{lem:wterminate}
Algorithm~\ref{alg:2} terminates after $O(n^3)$ iterations.
\end{lemma}

\begin{proof}[Proof of Lemma~\ref{lem:wpEF1inG}]
    By Lemma~\ref{lem:winitial}, it is guaranteed that there exists an optimal allocation to $\mathbf{LP}(\bw, \tau)$ satisfying (I'1) and (I'2).
    By Lemma~\ref{lem:wterminate}, Algorithm~\ref{alg:2} terminates after $O(n^3)$ iterations and outputs a wpEF1 optimal allocation to $\mathbf{LP}(\bw, \tau)$.
    This completes the proof.
\end{proof}

\section{Conclusion}\label{sec:con}
We study the problem of fairly and efficiently allocating indivisible chores among agents with additive cost functions. Our main contribution is to show the existence of an EF1 and PO allocation.
This is achieved by a novel combination of a fixed point argument and a discrete algorithm.

A major open question is whether one can design a polynomial-time algorithm to compute EF1 and PO allocations for both goods and chores. In our existence proof, the fixed point argument is used to guarantee an initial optimal allocation.
If such an allocation can be found in polynomial time, it would lead directly to a polynomial-time algorithm for computing an EF1 and PO allocation.

Another intriguing direction is to investigate the existence of EF1 and PO (or fPO) allocations under more general valuation classes beyond additive valuations, as well as in settings involving mixed division.

It is also an interesting question whether our techniques can be extended to settings with additional constraints, such as category constraints.

\appendix
\section{Missing Proofs}\label{ap:1}

In this appendix, we provide the detailed proofs omitted from the main text.

\subsection{Proofs from Section~\ref{sec:pre}}

\begin{proof}[Proof of Lemma~\ref{lem:non-degenerate}]
Assume for contradiction that $I^\epsilon$ is degenerate. 
Then, there exists a cycle $C = (i_1, j_1, i_2, j_2, \ldots, i_\ell, j_\ell, i_{\ell+1})$ in the complete bipartite graph $(N, M, E)$ such that $\pi(C) = 1$. 
Hence, we have
$$
A \cdot P^\epsilon = B \cdot Q^\epsilon,
$$
where $A = \prod_{k=1}^\ell c_{i_k j_k}$, $B = \prod_{k=1}^\ell c_{i_{k+1} j_k}$, and $P$ and $Q$ are products of distinct primes.

We first show that $A = B$. 
Suppose that $A > B$ (the case $A < B$ is analogous). Then, we have
$$
A \cdot P^\epsilon - B \cdot Q^\epsilon > A - B ( q_{nm})^{n\epsilon} > A - B - \frac{\delta'}{2} > 0,
$$
where the first inequality follows from $P^\epsilon > 1$ and $Q \le (q_{nm})^n$, 
the second from the inequality $B ((q_{nm})^{n\epsilon} - 1) \le  (c_{\max})^n ((q_{nm})^{n\epsilon} - 1) < \frac{\delta'}{2}$, 
and the third from the definition of $\delta'$.

This contradicts the assumption $A \cdot P^\epsilon = B \cdot Q^\epsilon$, hence we have $A = B$. 
However, this implies $P = Q$, contradicting the fact that $P$ and $Q$ are products of distinct primes. 
Therefore, $I^\epsilon$ is non-degenerate.
\end{proof}

\begin{proof}[Proof of Lemma~\ref{lem:ef1}]
Let $i \in N$ and $S, T \subseteq M$ such that $c_i(S) > c_i(T)$. 
Then we have
$$
c^{(\epsilon)}_i(S) - c_i(S) = \sum_{j \in S} (q^\epsilon_{ij} - 1)  c_{ij} > 0,
$$
and
$$
c^{(\epsilon)}_i(T) - c_i(T) = \sum_{j \in T} (q^\epsilon_{ij} - 1) c_{ij} 
\le (q^\epsilon_{nm} - 1)  m  c_{\max} < \frac{\delta}{2}.
$$
Hence,
$$
c^{(\epsilon)}_i(S) - c^{(\epsilon)}_i(T) > c_i(S) - c_i(T) - \frac{\delta}{2} > 0.
$$
Now, suppose that $\bx$ is not EF1 in $I$. 
Then, there exists a pair of agents $i, i' \in N$ and a chore $j \in \bx_i$ such that 
$c_i(\bx_i \setminus \{j\}) > c_i(\bx_{i'})$. 

By the argument above, it follows that
$c^{(\epsilon)}_i(\bx_i \setminus \{j\}) > c^{(\epsilon)}_i(\bx_{i'})$, which implies that
$\bx$ is not EF1 in $I^{\epsilon}$. 
\end{proof}

\begin{proof}[Proof of Lemma~\ref{lem:po}]
Let $\bx$ be an fPO allocation in $I^{\epsilon}$. By Lemma~\ref{lem:cha}, there exists $\bw = (w_i)_{i \in N} \in \mathbb{R}_{>0}^n$ such that $\bx$ is an optimal allocation of \textbf{LP}($\bw$). Let $i_1, \ldots, i_n$ be a permutation of $\{1, \ldots, n\}$ such that $w_{i_1} \ge w_{i_2} \ge \cdots \ge w_{i_n}$. Since scaling the weights does not affect the set of optimal solutions, we may, without loss of generality, normalize the weights so that $w_{i_1} = 1$.

We first show that we can assume 
$\frac{w_{i_{k+1}}}{w_{i_k}} \ge \frac{\delta}{q_{nm} c_{\max}}$ for all $k \in [n-1]$. Suppose for contradiction that there exists $k \in [n-1]$ such that $\frac{w_{i_{k+1}}}{w_{i_k}} < \frac{\delta}{q_{nm}  c_{\max}}$. Then, for any $i \in \{i_1, \dots, i_k\}$ and any chore $j$, we have:
\begin{align*}
w_{i_{k+1}}  c^{(\epsilon)}_{i_{k+1}j} 
&< \frac{w_{i_k} \delta   c_{i_{k+1}j}  q^{\epsilon}_{i_{k+1}j}}{q_{nm}  c_{\max}} \notag \\
&< w_{i_k}  \delta \notag \\
&< w_i  c_{ij}q^{\epsilon}_{ij}\\
&= w_i  c^{(\epsilon)}_{ij}, \notag
\end{align*}
where the second inequality follows from $\frac{c_{i_{k+1}j}  q^{\epsilon}_{i_{k+1}j}}{q_{nm}  c_{\max}} < 1$, noting that $\epsilon < 1$.
This implies that agents $i_1, \dots, i_k$ do not receive any chore in $\bx$. 
Otherwise, reallocating any chore from $i$ to $i_{k+1}$ decreases the objective value of $\mathbf{LP}(\bw)$, contradicting the optimality of $\bx$.
Therefore, we can scale $w_{i_{k+1}}, \dots, w_{i_n}$ so that $\frac{w_{i_{k+1}}}{w_{i_k}} = \frac{\delta}{q_{nm} c_{\max}}$ without affecting the optimality of $\bx$. By repeatedly applying this operation as long as there exists $k \in [n-1]$ such that 
$\frac{w_{i_{k+1}}}{w_{i_k}} < \frac{\delta}{q_{nm} c_{\max}}$,  we may assume that
$\frac{w_{i_{k+1}}}{w_{i_k}} \ge \frac{\delta}{q_{nm} c_{\max}}$ for all $k \in [n-1]$. Consequently, we obtain
$$
w_{i_n} \ge \left( \frac{\delta}{q_{nm} c_{\max}} \right)^{n-1}.
$$

Let $\bp = (p_j)_{j \in M}$ be an optimal solution of \textbf{Dual-LP}($\bw$). From dual feasibility, we have
$$
p_j \le w_i c^{(\epsilon)}_{ij} \le w_i c_{ij} q_{nm}^{\epsilon} < w_i c_{ij}(1+\alpha).
$$
By the duality theorem, we obtain
\begin{equation}
\bp(M) = \sum_{i \in N} w_i c^{(\epsilon)}_i(\bx_i) > \sum_{i \in N} w_i c_i(\bx_i). \label{eq:1}
\end{equation}
Suppose for contradiction that $\bx$ is not PO in $I$. Then, there exists another allocation $\by$ such that $c_i(\by_i) \le c_i(\bx_i)$ for all $i \in N$, and $c_h(\by_h) + \frac{\delta}{2} < c_h(\bx_h)$ for some $h \in N$. Therefore,
\begin{align}
\sum_{i \in N} w_i c_i(\bx_i)
&> w_h c_h(\by_h) + \frac{w_h \delta}{2} + \sum_{i \ne h} w_i c_i(\by_i) \notag \\
&> \frac{\bp(\by_h)}{1+\alpha} + \frac{w_h \delta}{2} + \sum_{i \ne h} \frac{\bp(\by_i)}{1+\alpha} \notag \\
&= \frac{\bp(M)}{1+\alpha} + \frac{w_h \delta}{2}. \label{eq:2}
\end{align}

Combining \eqref{eq:1} and \eqref{eq:2} yields:
\[
\frac{2\alpha}{w_h \delta} \bp(M) > 1 + \alpha.
\]

On the other hand, since $p_j < (1 + \alpha) c_{\max}$ and $w_h \ge \left( \frac{\delta}{q_{nm} c_{\max}} \right)^{n-1}$, we have:
\begin{align}
\frac{2\alpha}{w_h \delta} \bp(M)
&< \frac{2\alpha m}{w_h \delta} (1 + \alpha) c_{\max} \notag \\
&\le \frac{2m (q_{nm})^{n-1} (c_{\max})^{n}}{{\delta}^{n}} \cdot \alpha(1 + \alpha) \notag \\
&= 1 + \alpha, \notag
\end{align}
which is a contradiction. Hence, $\bx$ is PO in $I$.
\end{proof}

\subsection{Proofs from Section~\ref{sec:weight}}
\begin{proof}[Proof of Observation~\ref{ob:wpEF1}]
From the definition, we directly obtain 
    \begin{align*}
    \text{An allocation}~\bx~\text{is wpEF1}
     &\Longleftrightarrow \forall i, i' \in N, \frac{\hat{\bp}(\bx_{i})}{\alpha_i} \le \frac{\bp(\bx_{i'})}{\alpha_{i'}}  \\
     &\Longleftrightarrow \max_{i \in N} \frac{\hat{\bp}(\bx_{i})}{\alpha_i} \le \min_{i' \in N} \frac{\bp(\bx_{i'})}{\alpha_{i'}}.
    \end{align*}
\end{proof}

\begin{proof}[Proof of Lemma~\ref{lem:wpEF1iswEF1}]
Let $i, i' \in N$ be any pair of agents such that $i$ weighted envies $i'$.  
Since $i$ weighted envies $i'$, agent $i$ must have at least one chore.
Since the allocation $\bx$ is wpEF1, there exists a chore $j \in \bx_i$ such that $\frac{\bp(\bx_i \setminus \{j\})}{\alpha_i} \le \frac{\bp(\bx_{i'})}{\alpha_{i'}}.$
Since $\bx$ and $\bp$ are optimal solutions to the primal and dual problems, they satisfy the complementary slackness conditions by the complementary slackness theorem. Hence, we have
\begin{align*}
\frac{w_i c_i(\bx_i \setminus \{j\})}{\alpha_i} &= \frac{\bp(\bx_i \setminus \{j\})}{\alpha_i} \\
&\le \frac{\bp(\bx_{i'})}{\alpha_{i'}} \\
&\le \frac{w_i c_i(\bx_{i'})}{\alpha_{i'}}, 
\end{align*}
where the last inequality follows from the dual feasibility condition.

Therefore, we obtain $\frac{c_i(\bx_i \setminus \{j\})}{\alpha_i} \le \frac{c_i(\bx_{i'})}{\alpha_{i'}}$.
Since $i$ and $i'$ are arbitrary, it follows that $\bx$ is wEF1.
\end{proof}

\begin{proof}[Proof of Lemma~\ref{lem:wef1}]
Let $i \in N$ and $S, T \subseteq M$ such that $\frac{c_i(S)}{\alpha_i} > \frac{c_i(T)}{\alpha_{i'}}$. 
Then we have
$$
\frac{c^{(\epsilon)}_i(S) - c_i(S)}{\alpha_i} = \frac{1}{\alpha_i}\sum_{j \in S} (q^\epsilon_{ij} - 1)  c_{ij} > 0,
$$
and
$$
\frac{c^{(\epsilon)}_i(T) - c_i(T)}{\alpha_{i'}} = \frac{1}{\alpha_{i'}}\sum_{j \in T} (q^\epsilon_{ij} - 1) c_{ij} 
\le \frac{(q^\epsilon_{nm} - 1)  m  c_{\max}}{\alpha_{\min}} < \frac{\delta''}{2}.
$$
Hence,
$$
\frac{c^{(\epsilon)}_i(S)}{\alpha_i} - \frac{c^{(\epsilon)}_i(T)}{\alpha_{i'}} > \frac{c_i(S)}{\alpha_{i}} - \frac{c_i(T)}{\alpha_{i'}} - \frac{\delta''}{2} > 0.
$$
Now, suppose that $\bx$ is not wEF1 in $I$. 
Then, there exists a pair of agents $i, i' \in N$ and a chore $j \in \bx_i$ such that 
$\frac{c_i(\bx_i \setminus \{j\})}{\alpha_i} > \frac{c_i(\bx_{i'})}{\alpha_{i'}}$. 

By the argument above, it follows that
$\frac{c^{(\epsilon)}_i(\bx_i \setminus \{j\})}{\alpha_i} > \frac{c^{(\epsilon)}_i(\bx_{i'})}{\alpha_{i'}}$, which implies that
$\bx$ is not wEF1 in $I^{\epsilon}$. 
\end{proof}

\begin{proof}[Proof of Lemma~\ref{lem:winitial}]
Consider the set of chores $\Gamma_H(N \setminus R)$.
We first show that these chores can be allocated to agents in $N \setminus R$ such that every agent receives chores whose total price is at least $r_{\max}$.

Let $i \in R$ be arbitrary. Since $\bw \in \Delta^{n-1}$ is a weight vector satisfying the condition stated in Lemma~\ref{lem:wexistw}, there exists an optimal allocation $\by$ to $\mathbf{LP}(\bw, \tau)$ such that agent $i$ is weighted price envy-free, i.e., $\frac{\bp(\by_i)}{\alpha_i} \le \frac{\bp(\by_{i'})}{\alpha_{i'}}$ for all $i' \in N$.
Since $i \in R$, we have $r_{\max} \le \frac{\bp(\by_i)}{\alpha_i}$, it follows that $\frac{\bp(\by_{i'})}{\alpha_{i'}} \ge r_{\max}$ for all $i' \in N \setminus R$.
If there exists a chore in $\Gamma_H(N \setminus R)$ that is not allocated to $N \setminus R$ under $\by$, then we can reassign it arbitrarily to an agent in $N \setminus R$ without decreasing the total price of any bundle held by agents in $N \setminus R$. Thus, we can ensure that each agent in $N \setminus R$ receives a bundle from $\Gamma_H(N \setminus R)$ whose total price normalized by their weights is at least $r_{\max}$.

Next, consider the remaining chores $M_H \setminus \Gamma_H(N \setminus R)$ (which may be empty). By Lemma~\ref{lem:covermatching}, there exists a matching that covers $M_H \setminus \Gamma_H(N \setminus R)$.
By the definition of $\Gamma_H(N \setminus R)$, each of these chores is matched only to agents in $R$. Thus, we can assign them to agents in $R$ according to a matching.

Combining these assignments yields an optimal allocation $\bx$ to $\mathbf{LP}(\bw, \tau)$ that satisfies both \textup{(I1)} and \textup{(I2)}.
\end{proof}

\begin{proof}[Proof of Lemma~\ref{lem:wpathexist}]
For the first statement, observe that (I'2) implies that every agent in $N \setminus R$ must have at least one chore from $M_H$ in $\bx$. Since Lemma~\ref{lem:n-1} guarantees $|M_H| \leq n - 1$, there must exist at least one agent in $R$ who receives no chore from $M_H$,  
which implies that $U \neq \emptyset$.

For the second statement, since $\bx$ is not wpEF1, it follows from Observation~\ref{ob:wpEF1} and (I'2) that 
$$
\max_{i \in N} \frac{\hat{\bp}(\bx_i)}{\alpha_i} > \min_{i' \in N} \frac{\bp(\bx_{i'})}{\alpha_{i'}} \ge r_{\max},
$$
which implies $V \neq \emptyset$. Moreover, by (I'1) and (I'2), for any $i \in R$, we have $\frac{\hat{\bp}(\bx_i)}{\alpha_i} \le r_{\max} \le \min_{i' \in N} \frac{\bp(\bx_{i'})}{\alpha_{i'}}$, which implies that $i \notin V$. Hence, $V \subseteq N \setminus R$.

For the third statement, pick any $i_1 \in V$. By (I'2), agent $i_1$ must have at least one chore from $M_H$, hence there exists $j_1 \in \bx_{i_1}$ with $(i_1, j_1) \in E_H$. Since $d_H(j_1) \ge 2$, there exists another agent $i_2 \neq i_1$ with $(i_2, j_1) \in E_H$. If $i_2 \in U$, then we have found an alternating path from $U$ to $V$. Otherwise, $i_2 \in R \setminus U$ or $i_2 \in N \setminus R$. In either case, by the definition of $U$ and (I'2), there exists $j_2 \in \bx_{i_2}$ with $(i_2, j_2) \in E_H$.
By continuing inductively in this way, we construct an alternating path. Since the number of vertices is finite and $H$ is a forest, this process must eventually reach some agent in $U$. This completes the proof.
\end{proof}

\begin{proof}[Proof of Lemma~\ref{lem:winvariant}]
We prove by induction on $t$ that both invariants (I'1) and (I'2) are maintained in each iteration of the algorithm.
The initial optimal allocation $\bx^1$ satisfies (I'1) and (I'2) by Lemma~\ref{lem:winitial}.

Now assume that at time step $t$, the current allocation $\bx^t$ satisfies \textup{(I'1)} and \textup{(I'2)}.
Let $P=(i_0, j_1, i_1, j_2,\ldots ,i_{\ell-1}, j_\ell, i_\ell)$ be the shortest alternating path in $H$ computed in Algorithm~\ref{alg:2} during this iteration.
We show that after the transfer step, the updated allocation $\bx^{t+1}$ also satisfies \textup{(I'1)} and \textup{(I'2)}. 

\smallskip
\noindent\textbf{Invariant (I'1):} 
Fix any agent $i \in R$.
We consider the following four cases:
\begin{itemize}
    \item $\bx^{t+1}_i = \bx^t_i$.\\
    In this case, (I'1) clearly holds.
    \item $i = i_\ell$.\\
    By Lemma~\ref{lem:wpathexist}, we have $i_\ell \in V^t \subseteq N \setminus R$, and hence $i \neq i_\ell$. Thus, this case does not occur.
    \item $i = i_a$.\\
    In this case, agent $i$ must belong to $U^t$. Indeed, if $i \in R \setminus U^t$, then we have  $\frac{\bp(\bx^t_{i_a} \cup \{j_{a+1}\} \setminus \{j_a\})}{\alpha_{i_a}} \ge r_{\max}$ since $i\in R$, contradicting the choice of $i_a$. Thus, (I'1) is preserved after the update.
    \item Otherwise.\\
    Agent $i$ acquires one chore from $M_H$ and relinquishes one chore from $M_H$. Hence, the number of chores from $M_H$ assigned to $i$ remains unchanged, and (I'1) is preserved.
\end{itemize}

\smallskip
\noindent\textbf{Invariant (I'2):} 
Fix any agent $i \in N$.
We consider the following four cases:
\begin{itemize}
    \item $\bx^{t+1}_i = \bx^t_i$.\\
    In this case, (I'2) clearly holds.
    \item $i = i_\ell$.\\
    Since $i \in V^t$, we have 
    $$
    \frac{\bp(\bx^{t+1}_i)}{\alpha_i} = \frac{\bp(\bx^t_i \setminus \{j_\ell\})}{\alpha_i} \ge \frac{\hat{\bp}(\bx^t_i)}{\alpha_i} > r_{\max}.
    $$
    Hence, (I'2) continues to hold after the update.
    \item $i = i_a$.\\
    We have
    $$
    \frac{\bp(\bx^{t+1}_i)}{\alpha_i} = \frac{\bp(\bx^t_i \cup \{j_{a+1}\})}{\alpha_i} > \frac{\bp(\bx^t_i)}{\alpha_i} \ge r_{\max}.
    $$
    Thus, (I'2) continues to hold after the update.
    \item Otherwise.\\
    Let $i = i_k$ for some $a < k < \ell$ along the path $P$. 
    Agent $i_k$ acquires $j_{k+1}$ and relinquishes $j_k$. 
    By the choice of index $a$, the total prices of the updated bundle satisfies
    $$
    \frac{\bp(\bx^{t+1}_{i_k})}{\alpha_{i_k}} = \frac{\bp(\bx^t_{i_k} \cup \{j_{k+1}\} \setminus \{j_k\})}{\alpha_{i_k}} \ge r_{\max},
    $$
    otherwise $i_k$ would have been chosen as $i_a$, contradicting the definition of $a$. 
    Hence, (I'2) continues to hold after the update.
\end{itemize}

Therefore, both invariants (I'1) and (I'2) are maintained in each iteration of the algorithm.
\end{proof}

\section*{Acknowledgments.}
The author is grateful to Ayumi Igarashi for pointing out the recent paper~\cite{igarashi2025fair}, for early discussions on the use of weights in the proof of the existence of EF1+PO allocations, and for verifying the correctness of the proof.
This work was partially supported by the joint project of Kyoto University and Toyota Motor Corporation, titled ``Advanced Mathematical Science for Mobility Society'' and supported by JST ERATO Grant Number JPMJER2301, and JSPS KAKENHI Grant Number JP23K19956, Japan.

\bibliographystyle{plain}
\bibliography{main}
\end{document}